\documentclass{article}

\usepackage{amsmath,amssymb,latexsym,graphicx,enumerate,amsthm}

\newtheorem{theorem}{Theorem}[section]

\newtheorem{lemma}{Lemma}[section]
\newtheorem{definition}{Definition}[section]

\newcommand{\iif}{\mathtt{if}}

\newcommand{\puse}{\mathit{use}}

\newcommand{\pentry}{\mathit{entry}}
\newcommand{\pexit}{\mathit{exit}}

\newcommand{\CD}{\mathit{CD}}
\newcommand{\LIDD}{\mathit{LIDD}}
\newcommand{\LCDD}{\mathit{LCDD}}

\newcommand{\val}{\mathit{val}}
\newcommand{\store}{\mathit{store}}
\newcommand{\state}{\mathit{state}}
\newcommand{\Se}{\mathcal{E}}
\newcommand{\BT}{\mathbf{T}}
\newcommand{\BF}{\mathbf{F}}
\newcommand{\DefOrd}{\mathit{DefOrd}}
\newcommand{\ct}{\mathit{ct}}
\newcommand{\cf}{\mathit{cf}}
\newcommand{\econf}{\mathit{econf}}

\newcommand{\act}{\mathit{act}}

\newcommand{\ec}{\mathit{ec}}

\newcommand{\Next}{\mathit{Next}}
\newcommand{\iret}{\mathtt{ret}}
\newcommand{\avail}{\mathit{avail}}

\newcommand{\unchk}{\mathit{unchk}}
\newcommand{\chk}{\mathit{chk}}
\newcommand{\av}{\mathit{av}}
\newcommand{\udav}{\mathit{udav}}
\newcommand{\udec}{\mathit{udec}}
\newcommand{\mca}{\mathit{mca}}
\newcommand{\pcondC}{\mathit{condC}}
\newcommand{\pcondF}{\mathit{condF}}
\newcommand{\pcondL}{\mathit{condL}}
\newcommand{\pcondD}{\mathit{condD}}
\newcommand{\pcondCFLD}{\mathit{condCFLD}}

\newcommand{\Var}{\mathit{Var}}

\newcommand{\UR}{\mathit{UR}}

\newcommand{\Store}{\mathit{Store}}
\newcommand{\Val}{\mathit{Val}}
\newcommand{\Expr}{\mathit{Exp}}
\newcommand{\State}{\mathit{State}}
\newcommand{\Avail}{\mathit{Avail}}
\newcommand{\Econf}{\mathit{Econf}}

\begin{document}
\title{Semantical Equivalence of the Control Flow Graph and the Program
Dependence Graph}
\author{Sohei Ito\thanks{Contact: \texttt{ito@fish-u.ac.jp}} \\
{\small National Fisheries University, Japan}}
\date{}
\maketitle

\begin{abstract}
The program dependence graph (PDG) represents data and control dependence
 between statements in a program.
This paper presents an operational semantics of program dependence
 graphs.
Since PDGs exclude artificial order of statements that resides in
 sequential programs, executions of PDGs are not unique.
However, we identified a class of PDGs that have unique final states of
 executions, called \emph{deterministic PDGs}.
We prove that the operational semantics of control flow graphs is
 equivalent to that of deterministic PDGs.
The class of deterministic PDGs properly include PDGs obtained
 from well-structured programs.
Thus, our operational semantics of PDGs is more general than that of
 PDGs for well-structured programs, which are already established in
 literature.
\end{abstract}

\section{Introduction}
The program dependence graph (PDG)
\cite{kuck81dependence,FerranteOW87,horwitz89integrating} is a kind of
an intermediate representation of programs.
It is a directed graph whose nodes are program statements and edges
represent data dependence or control dependence between them.

The PDG is a useful representation for compiler code optimisation,
because many optimisation techniques are based on analyses
of data dependence and control dependence between statements in a program
\cite{FerranteO83,FerranteOW87}.
Using PDGs, we can simplify program optimisations, because many
optimisations can be carried out by simply scanning PDGs.
If optimisation techniques are applied to control flow graphs (CFGs), we
need to re-analyse dependence relationships in the program because
optimisation changes them.
However, if we apply optimisation techniques to PDGs, we do not need to
re-analyse dependence relationships because optimisation itself changes
them in the PDG.
Furthermore, since the PDG excludes artificial order of statements,
it is a suitable representation for vectorisation and parallelisation
\cite{warren84hierarchical,FerranteOW87,baxter89program}.

Although the PDG is a suitable representation for program optimisation,
the correctness of such optimisation techniques on PDGs is not
well-studied due to the lack of research on semantics of the PDG.

There are a few researches on the semantics of PDGs
\cite{selke89rewriting,cartwright89semantics,ramalingam89semantics}.
Selke \cite{selke89rewriting} introduced a rewriting semantics of a
PDG and proved that it is equivalent to the program semantics.
Cartwright and Felleisen \cite{cartwright89semantics} introduced a
denotational semantics of a PDG.
They proved that their semantics coincides with that of terminating
programs, and also proved that even a PDG corresponding to a
non-terminating program can return a result.
Ramalingam and Reps \cite{ramalingam89semantics} presented a formal
semantics of program representation graphs (PRG) that are extension of
PDGs.
A PRG has so-called $\phi$ nodes that determine which definition will
be used when the same variable is defined in multiple assignment
statements.
Thanks to $\phi$ nodes, PRGs can be naturally interpreted as data-flow
programs.
They define a semantics of a PRG as a function that maps an initial store
to a sequence of values for each node, and is given as the least fixed
point of a set of mutually recursive equations.

The above semantics of PDGs and PRGs cover only well-structured
programs.
Conditionals are restricted to the form $(\textbf{if}~e~\textbf{then}
\dots \textbf{else} \dots)$ and loops are restricted to the form
$(\textbf{while}~e\dots)$.
However, not all programs that appear during compiler code optimisation
fall into this type of programs.
Let us consider a program in Fig. \ref{fig:loop_while} and assume that
the expression \verb|x*y| is loop-invariant.
Since this loop has a conditional that determines iteration, we cannot
directly move the computation of \verb|x*y| outside of the loop.
Therefore, we need to transform the loop and move the computation of
\verb|x*y| before the loop as in Fig. \ref{fig:loop_while_mod}
\cite{nakata99eng}.
This loop has the form $\textbf{do} \dots \textbf{while}$ and is left
out of consideration in the existing semantics of PDGs.
Therefore, existing semantics of PDGs are not sufficient to prove
correctness of this optimsation technique.

\begin{figure}[tb]
\begin{center}
\includegraphics[scale=0.5]{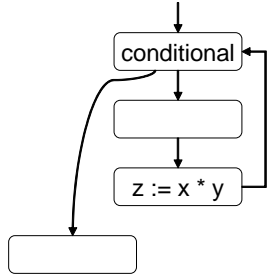}
\end{center}
\caption{A loop containing a loop invariant expression.}
\label{fig:loop_while}
\end{figure}

\begin{figure}[tb]
\begin{center}
\includegraphics[scale=0.5]{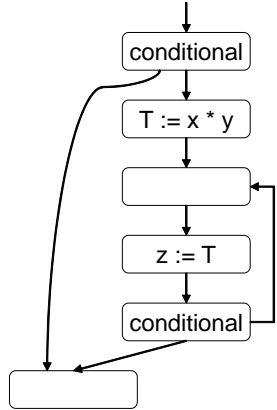}
\end{center}
\caption{After appling loop invariant code motion.}
\label{fig:loop_while_mod}
\end{figure}

In this paper, we present an operational semantics of PDGs that have
more complex control flow structures than that of well-structured
programs.
In our operational semantics, statements that are independent of each
other can be simultaneously executable.
Hence, there are many executions for a single PDG.
However, PDGs that satisfy certain structural constraints have the same
execution results.
We call this type of PDGs \emph{deterministic PDGs (dPDGs)}.

Moreover, we prove that our operational semantics of PDGs coincide with
that of CFGs for the class of dPDGs.
It gurantees that if programs are semantically equivalent to CFGs, they
are also semantically equivalent to PDGs.
This enables us to use PDGs as a foundation to discuss correctness of
several program transformations on PDGs.

This paper is organised as follows.
In section \ref{sec:cfg}, we introduce the CFG as a
presentation of programs in this paper and its operational semantics.
In addition to that, we define several types of dependence between
statements in a CFG.
In section \ref{sec:pdg}, we define the PDG.
Section \ref{sec:pdgsemantics} presents the operational semantics of
PDGs.
Furthermore, we define the class of deterministic PDGs and prove that
PDGs constructed from usual programs are deterministic PDGs.
In section \ref{sec:cfgpdgsem}, we prove that CFGs are semantically
equivalent to the corresponding PDGs.
The final section offers conclusion and future directions.

This paper is a revised version of the author's previous works
\cite{10.1007/978-3-540-77505-8_22,ito09operational} in that
definitions are revised and proofs are optimised and rectified.

\section{Control flow graph (CFG)} \label{sec:cfg}
We consider the CFG as a representation of programs.
A CFG is a pair $(N,E)$ of a set $N$ of nodes and a set $E \subseteq N
\times N$ of edges.
The nodes of a CFG are labelled by statements, and the edges of a CFG represent
 control flow of a program.
We only consider three types of statements; an assignment statement
($x:=e$), a conditional statement ($\iif~e$) and a return statement
($\iret~x $).
Assignment statements have exactly one successor and conditional statements
have exactly two successors.
The edges from a conditional statement are labelled differently from
each other, that is, either with true or false.
A CFG has a unique start node and a unique end node.
The start node has no predecessor, while the end node has no successor.
Any node in a CFG is reachable from the start node.
The last node must be a return statement.

Fig. \ref{fig:cfgex} is an example of a CFG.

\begin{figure}[tb]
\begin{center}
\includegraphics[scale = 0.5]{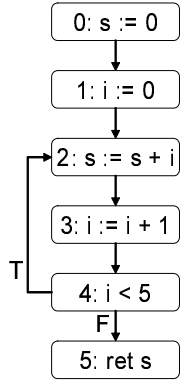}
\end{center}
\caption{An example of a CFG.}
\label{fig:cfgex}
\end{figure}

\begin{definition} \label{def:store}
Let $\Var$ be the set of variables and $\Expr$ be the set of
 expressions appearing in a CFG.
Let $\val$ be the set of values, which contains special values $\BT$ and $\BF$.
A \emph{store} is a function of the signature $\Var \longrightarrow \Val$.
We write $\Store$ for the set of stores.
An \emph{evaluation function} of expressions is $\Se : \Expr
 \longrightarrow (\Store \longrightarrow \Val)$.
\end{definition}

\begin{definition}(The operational semantics of CFG)
Let $G=(N,E)$ be a CFG.
A \emph{run} is a sequence of stores $\sigma_0 \stackrel{n_0}{\to}
 \sigma_1 \stackrel{n_1}{\to} \dots$ satisfing the following:
\begin{enumerate}
 \item $\sigma_0$ is a given initial store.
 \item $n_0$ is the initial node of $G$.
 \item $n_{i+1}$ is a successor of $n_i$ in $G$ and determined as
       follows:
 \begin{enumerate}
  \item If $n_i$ is an assignment statement, then $n_{i+1}$ is the
	unique successor of $n_i$.
  \item If $n_i$ is a conditional statement $(\iif~e)$, and $(n_i,p)_\BT
	\in E$ and $(n_i,q)_\BF \in E$, then
	\[
	n_{i+1} = 
	\begin{cases}
	 p & \mathrm{if} ~ \Se(e)\sigma_i = \BT\\
	 q & \mathrm{if} ~ \Se(e)\sigma_i = \BF
	\end{cases}
	\]
 \end{enumerate}
 \item $\sigma_{i+1}$ is determined as follows:
       \[
       \sigma_{i+1} = 
       \begin{cases}
	\sigma_i[x \mapsto \Se(e)\sigma_i] & \mathrm{if} ~ n_i = (x:=e) \\
	\sigma_i & \mathrm{otherwise}
       \end{cases},
       \]
where $\sigma[x \mapsto v]$ is the same as $\sigma$ except it maps $x$
       to $v$.
\end{enumerate}
\end{definition}

The following theorem is a trivial consequence of the definition.
\begin{theorem}
A run of a CFG is deterministic, that is, for any CFG $G$ and $\sigma_0$
 there exists the unique run of $G$.
\end{theorem}

There are several types of dependence between program statements.
In the following sections, we define control dependence, data dependence
and def-order dependence.

\subsection{Control flow dependence}
To define control dependence, we first define the notion of
\emph{post-domination}.
We add $\pentry$ node and $\pexit$ node to a CFG.
Entry node is the predecessor of the start node and exit node is
the successor of the end node.
Entry node has true edge to the start node and false edge to
exit node.
We call the resulting graph the \emph{augmented CFG} of an original CFG.
Figure \ref{fig:additionalcfg} is the augmented CFG for the graph in
Figure \ref{fig:cfgex}.

A \emph{path} is a sequence of nodes $n_0 n_1 n_2 \dots$ where $\forall
i \ge 0. (n_i, n_{i+1}) \in E$.
The \emph{length} of a path is the length of the sequence.

\begin{definition}[Post-domination]
Let $s$ and $t$ be nodes in a CFG $G$.
We say $t$ \emph{post-dominates} $s$ iff every path from $s$ contains
 $t$, and $t$ \emph{strictly post-dominates} $s$ iff $t$ post-dominates
 $s$ and $s \neq t$.
\end{definition}

\begin{definition}[Control dependence]
Let $s$ and $t$ be nodes in a CFG $G$.
We say $t$ is \emph{control-dependent} on $s$ iff
there exists a path of length greater than 1 from $s$ to $t$ such that
 $t$ post-dominates all nodes on that path except $s$, and $t$ does not
 strictly post-dominate $s$.
\end{definition}

The above definitions are slightly modified from usual post-domination and
control dependence.
The definition of post-domination and control dependence in this paper
are called ``strong forward domination'' and ``weak control dependence''
\cite{podgurski90formal}.
The usual definition is ``$s$ \emph{post-dominates} $t$ if every path from
$t$ \underline{to the exit node} contains $s$''.
Our definition omits the underlined phrase.

\begin{figure}[tb]
\begin{center}
\includegraphics[scale = 0.5]{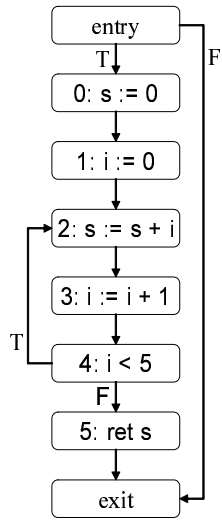}
\end{center}
\caption{A CFG augmented with entry node and exit node.}
\label{fig:additionalcfg}
\end{figure}

The difference between the usual definition and our definition appears
when a CFG contains a loop.
We illustrate this difference using an example CFG in
Fig. \ref{fig:additionalcfg}.
In this CFG, node 5 is not control-dependent on node 4 in the usual
definition since every path from node 4 to the exit node contains node
5, i.e. node 5 post-dominates node 4.
However, in our definition, node 5 is control-dependent on node 4 since
every path from node 4 does not contain node 5.
There exists a path iterating the loop infinitely.
Thus, node 5 does not post-dominate node 4.
Our definition reflects the intuition that whether node 5 is executed or not
 depends on the result of condition at node 4.

We discriminate two types of control dependence; true control dependence
and false control dependence.
Suppose $t$ is control-dependent on $s$. 
By definition, $s$ has two successors i.e. $s$ is a conditional statement.
Let $(s,u)_\BT \in E$ and $(s,v)_\BF \in E$.
If $u$ is post-dominated by $t$, $t$ is \emph{true control-dependent} on
$s$.
If $v$ is post-dominated by $t$, $t$ is \emph{false control-dependent}
on $s$.

We write $\CD(s,t)$ to represent that $t$ is control dependent on $s$.

\subsection{Data dependence}
There are two kinds of data dependence -- loop-independent data
dependence and loop-carried data dependence.
To define these types of dependence, we must define loops formally.
We cannot use the definition of \emph{natural loops} \cite{ASU86} since
we treat irreducible control flow graphs.

\begin{definition}[Loops, back edges]
Let $G$ be a directed graph.
A \emph{loop} is a strongly connected region in $G$.
A \emph{back edge} of a loop is an edge contained in that loop whose
 target node have an incoming edge from the outside of that loop.
\end{definition}

In Fig. \ref{fig:ex4}, the set of nodes $\{1,2,3,4,5,6\}$ comprises a
loop.
Back edges for this loop are $(3,4)$ and $(6,1)$.
In Fig. \ref{fig:loopex2}, the set of nodes $\{1,2,3,4,5\}$ and
$\{2,3,4\}$ comprise loops.
The back edges for $\{1,2,3,4,5\}$ and $\{2,3,4\}$ are respectively
$(5,1)$ and $(4,2)$.

\begin{figure}
\begin{minipage}{0.5\hsize}
\begin{center}
\includegraphics[height = 3cm]{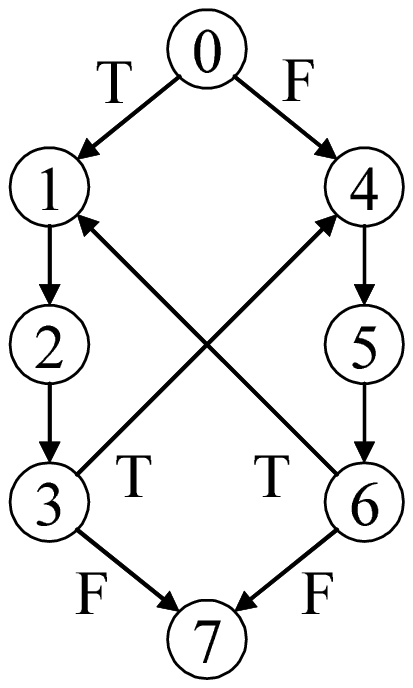}
\end{center}
\caption{An example of a loop.}
\label{fig:ex4}
\end{minipage}
\begin{minipage}{0.5\hsize}
\begin{center}
\includegraphics[height = 3cm]{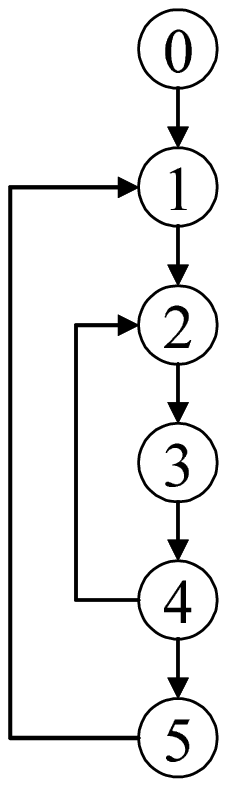}
\end{center}
\caption{An example of a nested loop.}
\label{fig:loopex2}
\end{minipage}
\end{figure}

\begin{definition}[Data dependence] \label{def:dd}
Let $G$ be a CFG. Let $s$ and $t$ be nodes in $G$.
We say \emph{the definition at $s$ reaches $t$} if there exists a
 variable $w$ such that $s$ defines $w$ and $t$ uses $w$ and there
 exists a path whose length is more than 1 in $G$ from $s$ to $t$ such
 that no node except $s$ on that path defines $w$.
We call such a path a \emph{reaching path} from $s$ to $t$.
We say $t$ is \emph{data-dependent} on $s$ if the definition at $s$
 reaches $t$.
There is a \emph{loop-independent data dependence} from $s$ to $t$ if
 $t$ is data dependent on $s$, and i) both $s$ and $t$ are not
 contained in the same loop or ii) if both $s$ and $t$ are contained in
 the same loop, there is a reaching path from $s$ to $t$ that does not
 contain any back edge of the loop.
There is a \emph{loop-carried data dependence} from $s$ to $t$ if both
 $s$ and $t$ are contained in the same loop, and there exists a reaching
 path from $s$ to $t$ that contains a back edge of the loop and
 there is no loop-independent data dependence from $s$ to $t$.
\end{definition}

We write $\LIDD(s,t)$ and $\LCDD(s,t)$ to represent there is a
loop-independent data dependence from $s$ to $t$ and there is a
loop-carried data dependence from $s$ to $t$, respectively.

\subsection{Def-order dependence}
Finally we define def-order dependence.

\begin{definition}[Def-order dependence]
Let $G$ be a CFG. Let $s$ and $t$ be nodes in $G$.
We say $t$ is \emph{def-order-dependent} on $s$ if there exists a node
 $u$ which is loop-independently data dependent on both $s$ and $t$, and
 both $s$ and $t$ define the same variable, and i) there exists a path
 from $s$ to $t$ and both $s$ and $t$ are not contained in the same
 loop, or ii) both $s$ and $t$ are contained in the same loop and $t$ is
 reachable from $s$ without passing any back edge of the loop, or iii)
 if both $s$ and $t$ are contained in the same loop, and $s$ is
 reachable from $t$ and $t$ is also reachable from $s$, then such paths
 must contain back edges of the loop.
\end{definition}

Note that in the case of iii), there is also a def-order dependence from
$t$ to $s$.
We write $\DefOrd(s,t)$ to represent $t$ is def-order-dependent on $s$.

\section{Program dependence graph (PDG)} \label{sec:pdg}
The definition of the PDG in this paper depends on the one by Selke
\cite{selke89rewriting}.
A PDG is a quintaple $(N,C,F,L,D)$ where $N$ is a set of nodes, and $C$,
$F$, $L$ and $D$ sets of edges, that is, subsets of $N \times N$.
The sets of edges $C$, $F$, $L$ and $D$ respectively represent four
types of dependence: control flow dependence, loop-independent data
dependence, loop-carried data dependence and def-order dependence.
$C$-edges are divided into true edges and false edges.

The subgraph $(N,C)$ is called the \emph{control dependence graph (CDG)}
of a PDG $(N,C,F,L,D)$.

Nodes are labeled with statements as in CFGs. Statements are one of the
following: an assignment statement ($x:=e$), a conditional statement
($\iif~e$) or a return statement ($\iret~x$).

If a node is labeled with an assignment statement, then there is no
$C$-edges outgoing from that state.
If a node is labeled with a conditional statement, then there is no
outgoing $F$, $L$ and $D$-edges from that node.
If $(p,q) \in F$ or $(p,q) \in L$, there is a variable $w$ that $p$
defines and $q$ uses.
If $(p,q) \in D$ then there is a variable $w$ such that both $p$ and $q$
define it and there is a node $u$ such that $(p,u) \in F$ and $(q,u) \in
F$.
Nodes labeled with $\iret$ statements do not have any outgoing edges.

\subsection{Construction of the PDG from a CFG}
Let $G=(N,E)$ be a CFG.
The PDG corresponding to $G$ is a graph $(N \cup \{\pentry\},C,F,L,D)$,
where, supposing $s$ and $t$ range over $N \cup \{\pentry\}$,
\begin{align*}
C &= \{(s,t)~|~\CD(s,t) \} \\
F &= \{(s,t)~|~\LIDD(s,t) \} \\
L &= \{(s,t)~|~\LCDD(s,t) \} \\
D &= \{(s,t)~|~\DefOrd(s,t) \}
\end{align*}

Note that the node $\pentry$ is the root node of the CDG $(N\cup
\{\pentry\}, C)$.

Hereafter, we write $c(s,t)$ for an element of $C$, $f(s,t)$ for an
element of $F$, $l(s,t)$ for an element of $L$ and $d(s,t)$ for an
element of $D$.
Furthermore, we discriminate true control dependence and false control
dependence in $C$ by writing $ct(s,t)$ and $cf(s,t)$.
We also discriminate elements in $F$ and $L$ by explicating a related
variable like $f_x(s,t)$ and $l_x(s,t)$.
Fig. \ref{fig:pdgex} presents the PDGand its CDG corresponding to the
CFG in Fig. \ref{fig:cfgex}.

\begin{figure*}[t]
\begin{minipage}{0.5\hsize}
\begin{center}
\includegraphics[scale=0.4]{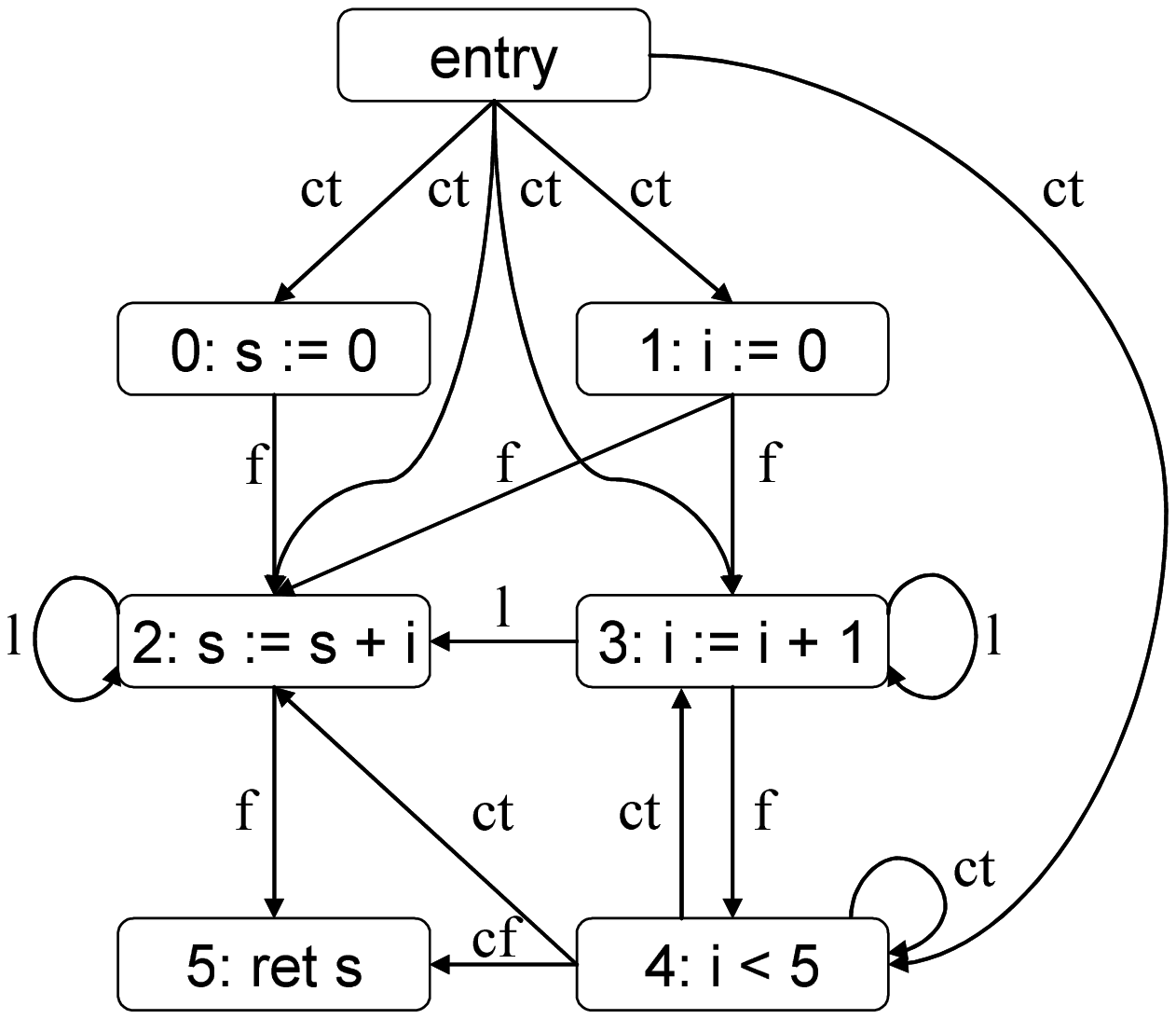}
\end{center}
\end{minipage}
\begin{minipage}{0.5\hsize}
\begin{center}
\includegraphics[scale=0.4]{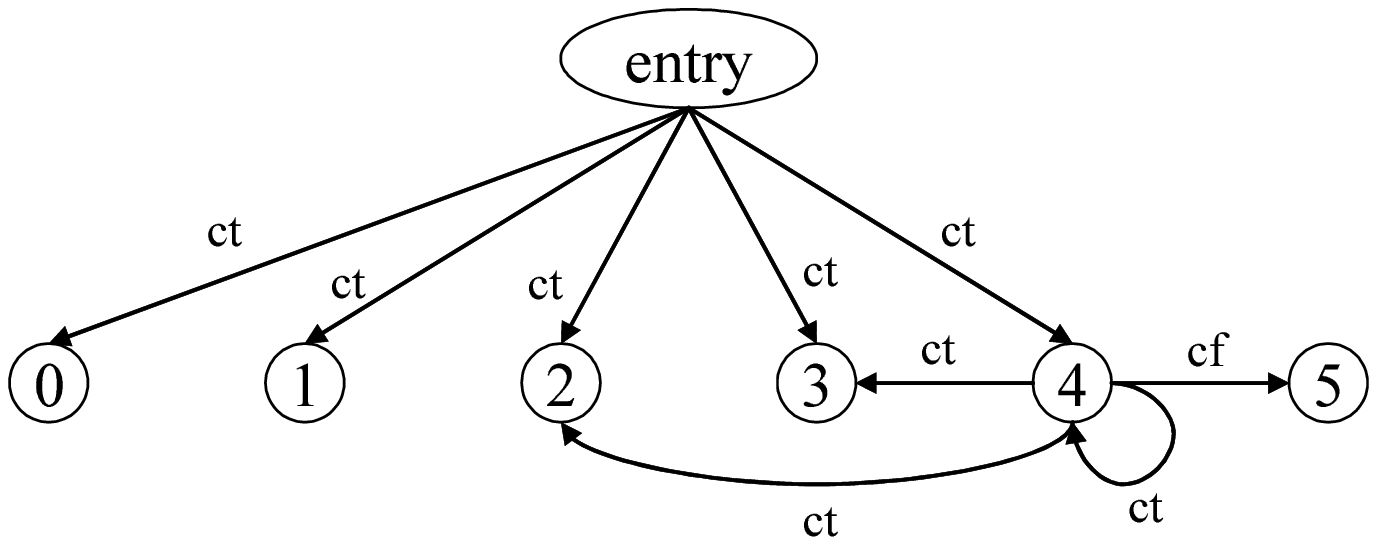}
\end{center}
\end{minipage}
\caption{The PDG and CDG corresponding to the CFG in Fig. \ref{fig:cfgex}}
\label{fig:pdgex}
\end{figure*}

\section{An operational semantics of the PDG} \label{sec:pdgsemantics}
In this section, we present an operational semantics of the PDG.
The basic idea is to define states in runs of PDGs as follows.
\[
 \state = (\avail, \econf)
\]

$\avail$ is a function that takes a node and a variable
and return a value.
Recall that $\store$ is a set of functions taking a variable and return
a value.
The difference between $\store$ and $\avail$ is that $\store$ is global
bindings of variables while $\avail$ is local bindings of variables for
each node.
We do not need global stores since an execution of a statement affects
only nodes that are dependent on it.

$\econf$ is a function that takes an edge and return the state of
the edge.
Functions in $\econf$ indicate whether nodes are executable or not.
For example, if node $n$ is an $\iif$ statement and the expression is
evaluated to $\BT$ then its outgoing $\ct$-edges are ``activated'' and
outgoing $\cf$-edges are ``inactivated.''
Whether a node is executable or not is determined by the state of its
incoming edges.
$\econf$ is used to represent a configuration of edges in executions
of PDGs.

To define the operational semantics of the PDG, we introduce the
following notions on PDGs.

\subsection{Preliminary}
In order to define the operational semantics of the PDG, we define
several notions related to it.

\begin{definition}[Looping edges]
Let $G=(N,C,F,L,D)$ be a PDG and $G_C = (N,C)$ be the CDG of $G$.
Let $R$ be a loop in $G_C$.
Suppose $\ct(p,q)$ (resp. $\cf(p,q)$) is a back edge of $R$.
The \emph{looping edges} of $R$ are all $\ct$-edges (resp. $\cf$-edges)
 from $p$ whose destinations have predecessors other than $p$.
\end{definition}

Note that the notions of loops and back edges are not restricted to
the CFG.
In the CDG in Figure \ref{fig:pdgex}, there is a loop which consists of
only node 4.
Its back edge is $\ct(4,4)$.
Therefore, its looping edges are $\ct(4,2), \ct(4,3)$ and $\ct(4,4)$.

We explain the correspondence between loops in a CFG and those in the
CDG.
Before that, we define \emph{iteration statements} of loops in the CFG.

\begin{definition} [Iteration statements] \label{def:iteration_statement}
An \emph{iteration statement} of a loop in a CFG is an $\iif$ statement
 in that loop that has outgoing edges to nodes both inside and outside
 of that loop.
It is trivial that for any loop in a CFG, there is at least one
 iteration statement.
\end{definition}

Compare the CFG in Fig. \ref{fig:ex4} with its corresponding CDG in
Fig. \ref{fig:subgex}.
We see that there is a loop $\{1,2,3,4,5,6\}$ in Fig. \ref{fig:ex4}
whose back edges are $(3,4)$ and $(6,1)$.
Therefore, node 3 and 6 are iteration statements in this program.
Meanwhile, the loop existing in Fig. \ref{fig:subgex} is $\{3,6\}$.
Thus, we can see that iteration statements in a CFG will be nodes in
a loop in the corresponding CDG.
Furthermore, the back edges of the loop $\{3,6\}$ in
Fig. \ref{fig:subgex} are $\ct(3,6)$ and $\ct(6,3)$.
Therefore, the looping edges are $\ct(3,4)$, $\ct(3,5)$, $\ct(3,6)$,
$\ct(6,1)$, $\ct(6,2)$ and $\ct(6,3)$.
If a node in a loop in a CDG is an iteration statement, then the
looping edges outgoing from that node are $C$-edges to the nodes that
will be executed when the loop iterates.

Intuitively, looping edges represent control dependence to the
statements that will be executed in the next step of the loop.

\begin{figure}[tb]
\begin{center}
\includegraphics[height=5cm]{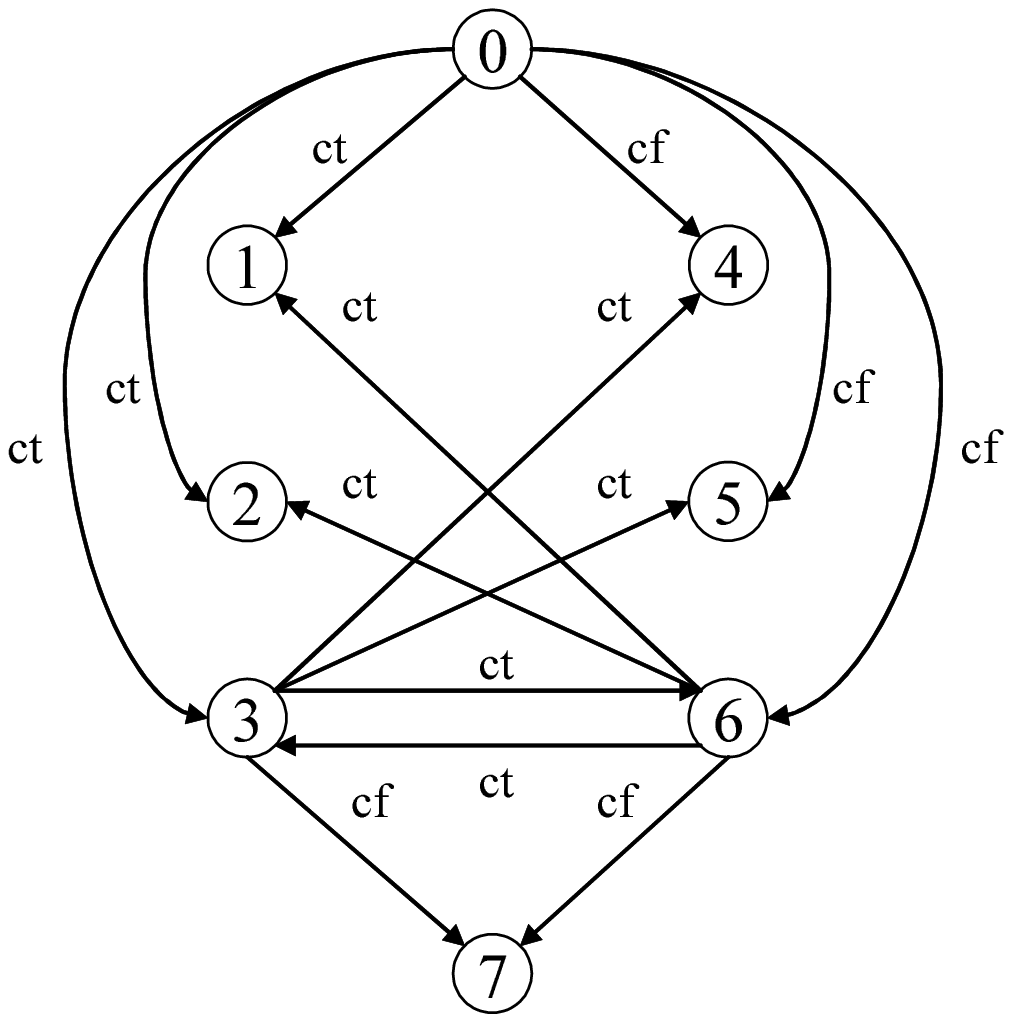}
\begin{align*}
G(0) &\quad \{1,2,3,4,5,6,7\} \\
G_T(0) &\quad \{1,2,3,7\} \\
G_F(0) &\quad \{4,5,6,7\} \\
G(3) &\quad \{4,5,6,7\} \\
G_T(3) &\quad \{4,5,6\} \\
G_F(3) &\quad \{7\} \\
G(6) &\quad \{1,2,3,7\} \\
G_T(6) &\quad \{1,2,3\} \\
G_F(6) &\quad \{7\}
\end{align*}
\end{center}
\caption{The CDG corresponding to Fig. \ref{fig:ex4} and its subgraphs.}
\label{fig:subgex}
\end{figure}

\begin{definition}
 Let $G=(N,C,F,L,D)$ be a PDG.
We define $\widehat{C}$ as the set obtained by removing all looping
 edges from $C$.
\end{definition}

\begin{definition}
 Let $p$ and $q$ be nodes in a PDG $(N,C,F,L,D)$.
We say $q$ is \emph{$C$-reachable} from $p$ if $q$ is reachable from $q$
 by $C$-edges.
We say $q$ is \emph{$\widehat{C}$-reachable} from $p$ if $q$ is
 reachable from $q$ by $\widehat{C}$-edges.
We respectively say a $C$-successor, a $ct$-successor and a
 $cf$-successor for a successor by a $C$-edge, a $ct$-edge and a
 $cf$-edge.
\end{definition}

\begin{definition}[Subgraphs] \label{def:subgraphs}
 Let $G = (N,C,F,L,D)$ be a PDG and $p \in N$.
$G(p) = (N',C',F',L',D')$ is the \emph{subgraph} of node $p$
 defined as follows:
\begin{align*}
N' &= \{n ~|~ \exists q. c(p,q) \in C \wedge n \textrm{ is
 $\widehat{C}$-reachable from $q$} \} \\
C' &= C \cap (N' \times N') \\
F' &= F \cap (N' \times N') \\
L' &= L \cap (N' \times N') \\
D' &= D \cap (N' \times N')
\end{align*}
$G(p)$ is the subgraph that consists of nodes reachable from $p$ by
 $C$-edges that are not looping edges.
Note that $c(p,q)$ itself may be a looping edge.
Similarly, we define $G_T(p)$ (resp. $G_F(p)$) as the subgraph consisting
 of the nodes $\widehat{C}$-reachable from $ct$-successors
 (resp. $cf$-successors) of $p$.
Formally, $G_T(p)=(N',C',F',L',D')$ is defined as:
\begin{align*}
N' &= \{n ~|~ \exists q. \ct(p,q) \in C \wedge n \textrm{ is
 $\widehat{C}$-reachable from $q$} \} \\
C' &= C \cap (N' \times N') \\
F' &= F \cap (N' \times N') \\
L' &= L \cap (N' \times N') \\
D' &= D \cap (N' \times N')
\end{align*}
The definition of $G_F(p)$ is obtained from replacing the condition
 $\ct(p,q) \in C$ with $\cf(p,q) \in C$.
We further define $G^*(p) = (N',C',F',L',D')$ as follows:
\begin{align*}
N' &= \{n ~|~ \exists q. c(p,q) \in C \wedge n \textrm{ is $C$-reachable
 from $q$} \} \\
C' &= C \cap (N' \times N') \\
F' &= F \cap (N' \times N') \\
L' &= L \cap (N' \times N') \\
D' &= D \cap (N' \times N')
\end{align*}
We can similarly define $G_T^*(p)$ and $G_F^*(p)$.
\end{definition}

Intuitively, $G(p)$ consists of such statements that are at the same
iteration level in the corresponding CFG.
$G_T(p)$ consists of statements that will be executed when $p$ is
evaluated to $\BT$, and $G_F(p)$ consists of statements that will be
executed when $p$ is evaluated to $\BF$.

Fig. \ref{fig:subgex} shows subgraphs of the CDG in Fig. \ref{fig:ex4},
where looping edges are $\ct(3,4)$, $\ct(3,5)$, $\ct(3,6)$, $\ct(6,1)$,
$\ct(6,2)$ and $\ct(6,3)$.

Using above notions, we define the operational semantics of PDGs.
\begin{definition}
An \emph{avail} and an \emph{econf} are respectively members of the
 following sets of functions $\Avail$ and $\Econf$:
\begin{gather*}
\Avail = N \times \Var \longrightarrow \Val \\
\Econf = C \oplus F \oplus L \oplus D \longrightarrow \{\chk, \unchk,
 \act \}.
\end{gather*}
Here $\oplus$ is the direct sum.
An evaluation function $\Se$ is defined as
\[
 \Se : N \times \Expr \longrightarrow (\Avail \longrightarrow \Val).
\]
A \emph{state} is a member of $\State= \Avail \times \Econf$.
\end{definition}

Although we use the same symbol $\Se$ as the evaluation function in
definition \ref{def:store}, they are easily discriniminated by the
signatures.


\begin{definition}[Executable nodes]
 We first introduce the following predicates.
Suppose $s=(\av,\ec)$, $\av \in \Avail$, $\ec \in \Econf$ and $n$ is a
 node.
\[
\begin{array}{lcl}
\pcondC(s,n) & \stackrel{\mathrm{def}}{=} & \exists c(p,n) \in
 C. \ec(c(p,n)) = \act \\
& & \wedge (\forall q \in G^*(n).q \neq n \Rightarrow \forall c(r,q) \in
 C. \ec(c(r,q)) \neq \act) \\
\pcondF(s,n) & \stackrel{\mathrm{def}}{=} & \forall f(p,n) \in
 F. \ec(f(p,n)) = \chk \\
\pcondL(s,n) & \stackrel{\mathrm{def}}{=} & \forall l(n,p) \in
 L. n\neq p \Rightarrow \ec(l(n,p)) = \chk \\
\pcondD(s,n) & \stackrel{\mathrm{def}}{=} & \forall d(p,n) \in
 D. \ec(d(p,n)) = \chk
\end{array}
\]
We write $\pcondCFLD(s,n)$ for $\pcondC(s,n) \wedge \pcondF(s,n) \wedge
 \pcondL(s,n) \wedge \pcondD(s,n)$.
We define function $\Next : \State \longrightarrow 2^N$ as:
\[
 \Next(s) = \{n ~|~ \pcondCFLD(s,n)\}.
\]
$\Next(s)$ represent the \emph{executable nodes} at state $s$.
\end{definition}

\begin{definition}[Update function of avails] \label{def:udav}
 We define function $\udav : N \times \Avail \longrightarrow \Avail$ as
 follows:
\[
\udav(n, \av) =
\begin{cases}
av[(p,x) \mapsto \Se(n, e)\av : (n,p) \in F \oplus L] \quad \mathrm{if }~ n = (x:=e) \\
av \quad \textnormal{otherwise}
\end{cases},
\]
where $\av[(p,x) \mapsto v : (n,p) \in F \oplus L]$ is the same as $\av$ except it maps
 $(p,x)$ to $v$ for each $p$ such that $(n,p) \in F \oplus L$.
\end{definition}

\begin{definition}[Update function of econfs] \label{def:udec}
 We define function $\udec : N \times \State \longrightarrow \Econf$ as
 follows:
\[
 \udec(n, (\av, \ec)) = \ec',
\]
where $\ec'$ is determined according to $n$, $\av$ and $\ec$ as follows.
\begin{itemize}
\item $n$ is any type.
\begin{itemize}
\item
$\ec'(c(p,n)) := \unchk$, if $c(p,n) \in C$.
\item
$\ec'(l(p,n)) := \chk$, if $l(p,n) \in L$.
\end{itemize}

\item $n = (x:=e)$.
\begin{itemize}
\item
$\ec'((n,p)) := \chk$, if $(n,p) \in F \oplus D$.
\end{itemize}
\item $n = (\iif~e)$ and $\Se(n, e)\av = \BT$.
\begin{itemize}
\item
$\ec'(ct(n, p)) := \act$, if $\ct(n,p) \in C$.
\item
$\ec'((q,r)) := \unchk$, if $q \in G_T(n) - \{n\}$ and $(q,r) \in
     F \oplus D$.
\item
$\ec'(l(r,q)) := \unchk$, if $q \in G_T(n) - \{n\}$, $r \in G_T(n)$ and
     $l(r,q) \in L$.
\item
$\ec'((q,r)) := \chk$, if $q \in G_F(n) - G_T(n)$ and $(q,r) \in F
     \oplus D$. 
\item
$\ec'(l(q,r)) := \unchk$, if $q \in G_F(n) - G_T(n)$, $r \not \in G_F(n)
     - G_T(n)$ and $l(q,r) \in L$.
\item
$\ec'(l(r,q)) := \chk$, if $q \in G_F(n) - G_T(n)$, $r \not \in G_F(n)
     - G_T(n)$ and $l(r,q) \in L$.
\end{itemize}
\item $n = (\iif~e)$ and $\Se(n, e)\av = \BF$．
This case is the same as the case $\Se(n, e)\av = \BT$.
Replace $\ct$ with $\cf$ and $G_T(n)$ with $G_F(n)$ in the definition.
\end{itemize}
$\ec$ and $\ec'$ coincides on the edges that are not modified according
 to the rule above.
\end{definition}

\begin{definition}[Operational semantics of the PDG] \label{def:pdgsemantics}
 Let $G=(N,C,F,L,D)$ be a PDG.
A \emph{run} of $G$ is a sequence of states $s_0 \stackrel{n_0}{\to} s_1
 \stackrel{n_1}{\to} \dots$, where $s_i \in \State$ and $n_i \in N$.
Let $s_i = (\av_i, \ec_i)$, $\av_0$ be a given initial avail and $\ec_0$
 be the initial econf defined as follows:
\begin{gather*}
\forall (p,q) \in C. \ec_0((p,q)) = 
\begin{cases}
\act & \textnormal{if}~p = \pentry \\
\unchk & \textnormal{otherwise}
\end{cases}, \\
\forall (p,q) \in F \oplus L \oplus D. \ec_0((p,q)) = \unchk
\end{gather*}
Then $n_i$ and $s_{i+1}$ are determined as follows:
\begin{align*}
s_i \stackrel{n_i}{\to} s_{i+1}
& \stackrel{\textrm{def}}{\Longleftrightarrow} n_i \in \Next(s_i) \wedge
 \\
& \av' = \udav(n_i,\av_i) \wedge \ec' = \udec(n_i,s_i).
\end{align*}
A run of a PDG is either finite or infinite.
If a run is finite and the
 last state is $s$, then $\Next(s)=\emptyset$.
\end{definition}

\subsection{Deterministic PDG} \label{subsec:dPDG}
A PDG does not necessarily have a unique run even if it starts with the
same initial state, because it can express programs more expressive than
usual sequential ones.
However, we can prove that for a specific class of PDGs, called
deterministic PDGs (dPDGs), it has runs that end with the same state
from the same initial state, if they are finite.
In this section, we define the class of dPDGs and prove some properties
of it.

We first introduce the notion of the \emph{minimal common ancestor (mca)}.

\begin{definition}[Minimal common ancestors]
 Let $G=(N,C,F,L,D)$ be a PDG and $p, q \in N$.
The \emph{minimal common ancestors} of nodes $p$ and $q$ are the last
 nodes of the longest common prefix of the acyclic paths from $\pentry$
 to $p$ and $q$ in $(N,C)$.
We write $\mca(p,q)$ for the set of minimal common ancestors of $p$ and
 $q$.
\end{definition}

\begin{definition}[Deterministic PDGs] \label{def:dPDG}
A PDG $G=(N,C,F,L,D)$ is a deterministic PDG if it satisfies the
 following three conditions where the term $\mathit{LP}$ denotes the set
 of loops in the CDG $(N,C)$:
\begin{enumerate}
\item $(p,n) \in C \wedge (q,n) \in C$ implies $p \in G^*(q)$ or $q \in
      G^*(p)$ or $\forall r \in \mca(p,q). \forall Q \in
      \{\BT,\BF\}. \{p,q\} \not \subseteq G_Q^*(r)$.
\item $f_x(p,u) \in F \wedge f_x(q,u) \in F \wedge \exists r \in
      \mca(p,q). \exists Q \in \{\BT,\BF\}. \{p,q\} \subseteq G^*_Q(r)$
      implies $d(p,q) \in D \vee d(q,p) \in D$．
\item $(p,q) \in F \oplus D$ implies $\forall R \in \mathit{LP}. \forall
      r \in R. p \in G(r) \Rightarrow q \in G^*(r)$．
\end{enumerate}
\end{definition}

Condition 1 states that if a node $n$ is control-dependent on both $p$
and $q$, then $p$ and $q$ cannot be executable simultaneously.
Fig. \ref{fig:dpdgcond1cdg} illustrates this condition.
The outgoing edges from any minimal common ancestor $r$ of $p$ and $q$
must have different labels between the one that leads to $p$ and the one
to $q$. Otherwise either $p$ must be controlled by $q$ or $q$ must be
controlled by $p$.

\begin{figure}[tb]
\begin{center}
\includegraphics[scale=0.5]{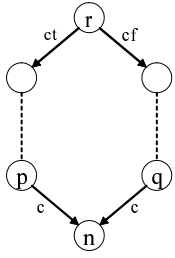}
\caption{Condition 1 of dPDG.}
\label{fig:dpdgcond1cdg}
\end{center}
\end{figure}

Condition 2 states that if $p$ and $q$ define the same variable $x$,
have a common $F$-successor and can be executable simultaneously, then
there must be a def-order dependence between $p$ and $q$.
This is illustrated in Fig. \ref{fig:dpdgcond3pdg}.
In this case, $p$ and $q$ can be simultaneously executable.
If there is no def-order dependence between $p$ and $q$, the value of
the variable at node $n$ can be different depending on which node
(i.e., $p$ or $q$) is executed first.

\begin{figure}[tb]
\begin{center}
\includegraphics[scale=0.5]{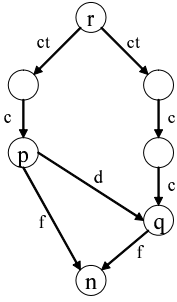}
\caption{Condition 2 of dPDG.}
\label{fig:dpdgcond3pdg}
\end{center}
\end{figure}

Condition 3 states that if there is data dependence from $p$ to $q$ and
$p$ is control-dependent on a node $r$ in some loop $R$, then $q$ must
be control-dependent on some node in that loop.
Fig. \ref{fig:dpdgcond4pdg} depicts the condition, where $r'$
corresponds to 'some node' in the previous sentence (note that $r' \in
G^*(r)$ because both are contained in the same loop).
In Fig. \ref{fig:dpdgcond4pdg}, if node $q$ is not contained in $G(r')$
then whether $q$ satisfies $pcondC$ can be determined independently on a
condition of $\iif$ statements in $R$.
Meanwhile, since $p$ is contained in $G(r)$, executability of $p$
depends on $r$.
Since $R$ is a loop, $r$ can be repeated, which implies that $p$ can be
repeated many times and the value of a variable used in $q$ may be
changed many times.
However, as stated above, $q$ can satisfy $\pcondC$ independently on the
loop.
Therefore, once $p$ is executed and data dependence to $q$ is satisfied,
$q$ can be executable.
Then the result of execution of the PDG depends on how many times $p$ is
repeated when $q$ is executed.

\begin{figure}[tb]
\begin{center}
\includegraphics[scale=0.5]{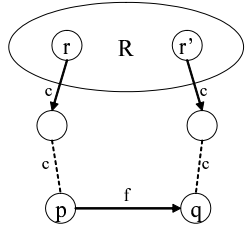}
\caption{Condition 3 of dPDG}
\label{fig:dpdgcond4pdg}
\end{center}
\end{figure}

Actually, there are PDGs constructed from CFGs that are not dPDGs.
CFGs like the one in Fig. \ref{fig:loopex4} do not satisfy condition
1 of dPDGs, where a back edge from some $\iif$ statement comes into a
branch.

\begin{figure}[tb]
\begin{center}
\includegraphics[scale=0.5]{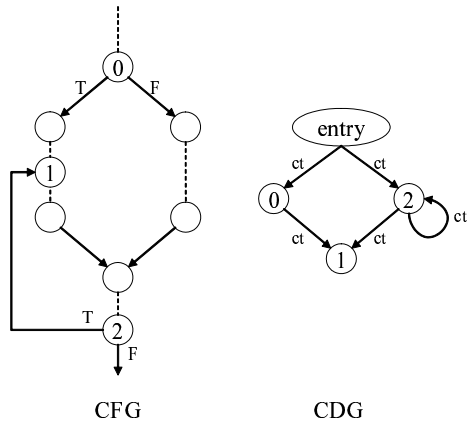}
\end{center}
\caption{An example CFG whose corresponding PDG is not dPDG.}
\label{fig:loopex4}
\end{figure}

As for Condition 2 and 3, we can prove that for all CFGs, their PDGs
satisfy them.

\begin{theorem}
A PDG constructed from a CFG satisfies Condition 2 and 3 of Definition
 \ref{def:dPDG}.
\end{theorem}
\begin{proof} We prove the claim separately.
 \begin{description}
  \item[Condition 2] In the corresponding PDG, suppose that for
	     statements $p, q$ and $u$ and variable $x$ we have
	     $f_x(p,u) \in F \wedge f_x(q,u) \in F \wedge \exists r \in
	     \mca(p,q). \exists Q \in \{\BT,\BF\}. \{p,q\} \subseteq
	     G^*_Q(r)$.
By the definition of loop-independent data dependence, $u$ is reachable
	     from both $p$ and $q$ in the CFG.
By the definition of control dependence, both $p$ and $q$ are reachable
	     from a $\BT$-successor of $r$ or a $\BF$-successor of $r$.
Since $r$ is a mca of $p$ and $q$, $p$ is reachable from $q$ or $q$ is
	     reachable from $p$ in the CFG.
Then, by the definition of def-order dependence, we have $\DefOrd(p,q)$
	     or $\DefOrd(q,p)$.
  \item[Condition 3] In the corresponding PDG, suppose that for
	     statements $p$ and $q$ we have $(p,q) \in F \oplus D$, and
	     for a node $r$ in some loop $R$ in the CDG, we have $p \in
	     G(r)$.
By Lemma \ref{lemma:CFGloopA} (below), a node in a loop in a CDG is an
	     iteration statement in a loop in the corresponding CFG.
Therefore, $r$ is an iteration statement in the CFG.
Since $(p,q) \in F \oplus D$, $q$ is reachable from $p$ in the CFG.
Moreover since $p \in G(r)$, $p$ is reachable from $r$ in the CFG.
Thus, $q$ is also reachable from $r$ in the CFG.
Since $r$ is an iteration statement of some loop in the CFG, by Lemma
	     \ref{lemma:CFGloopB} (below) we have $q \in G^*(r)$.
 \end{description}
\end{proof}

\begin{lemma} \label{lemma:CFGloopA}
 Let $P$ be a CFG and $G$ be the CDG of the corresponding PDG.
A node in a loop in $G$ is an iteration statement in a loop in $P$.
\end{lemma}
\begin{proof}
 Let $p$ and $q$ be nodes in a loop $R'$ in $G$.
By the definition of the loop, $p$ is reachable from $q$ and vice versa
 in $G$.
Since $p$ and $q$ have outgoing $C$-edges, they are $\iif$ statements.
By the definition of control dependence, if a node is reachable from
 another node in $G$, it is also reachable in $P$.
Thus, $p$ and $q$ are $\iif$ statements in some loop $R$ in $P$.
In other words, we showed that a node in a loop in $G$ is contained in
 some loop in $P$.
Then there exists at least one loop iteration statement in $R$ by
 Definition \ref{def:iteration_statement}.
We call it $r$.
Let $s \in R'$ be a statement such that $(s,r) \in C$.
By the definition of control dependence, $s$ is not post-dominated by
 $r$ in $P$, which implies $s$ has an outgoing edge that comes into a
 node inside $R$ and one that comes into a node outside $R$.
Thus, by the definition of the iteration statement, $s$ is an iteration
 statement.
Since $s$ is an iteration statement, any node $s' \in R'$ such that
 $(s',s) \in C$ is again an iteration statement by the above argument.
Thus, we proved that every node in a loop in $G$ is an iteration
 statement.
\end{proof}

\begin{lemma} \label{lemma:CFGloopB}
 Let $P$ be a CFG and $G$ be the CDG of the corresponding PDG.
If $p$ is a statement reachable from some node in a loop $R$ in $P$,
 then  for any iteration statement $r \in R$, we have $p \in G^*(r)$.
\end{lemma}
\begin{proof}
 If $p$ is a statement in $R$, we have $p \in G(r')$ for some iteration
 statement $r'$ of $R$.
If $p$ is not a statement in $R$, by the definition of control
 dependence, $p \in G^*(r')$ for some iteration statement $r'$ of $R$.
Here, by Lemma \ref{lemma:CFGloopA}, any iteration statement $r$
 contained in $R$ is contained in a loop in $G$ that contains $r'$.
Therefore, we have $r' \in G^*(r)$ and $p \in G^*(r)$ holds.
\end{proof}

\section{Semantical correspondence between CFG and PDG} \label{sec:cfgpdgsem}
In this section, we prove the semantical correspondence between a CFG
and its corresponding PDG.
Here the important assumption is that the PDG should be a dPDG.
Hence, our claim is that if the PDG constructed from a CFG is a dPDG, then
the results of the executions of the CFG and the PDG coincide.
This is derived from the following two properties:

\begin{description}
 \item[Property 1] There is the same execution order of statements for
	    both a CFG and its corresponding PDG (Theorem
	    \ref{thm:corresponding_run}).
 \item[Property 2] In a dPDG, every run starting with the same initial
	    state has the same final state (Theorem \ref{thm:dpdg}).
\end{description}

By these properties, if a PDG corresponding to a CFG is a dPDG, the
execution of the CFG and any execution of the PDG have the same final
result.

To prove these theorems, we need a number of lemmata.
To make the presentation easy and understandable, we introduce some
notations.
Hereafter, we write $P=(N,E)$ for a CFG, $G=(N \cup
\{\pentry\},C,F,L,D)$ for the corresponding PDG, $\sigma$ or $\sigma_i$
for stores and $s=(\av,\ec)$ or $s_i=(\av_i,\ec_i)$ for states.
If $Q \in \{\BT, \BF\}$, we write $\bar{Q}$ for the negation of $Q$.
We write $[i,j]$ for the set of integers from $i$ to $j$ and $(i,j)$ for
the set of integers from $i+1$ to $j-1$.

We define a relation that represents a state (i.e., store) of a
run of a CFG is `equal' to a state of a run of a PDG.
Intuitively, it says that a variable used at node $n$ is the same in
$\sigma$ and $\av$.

\begin{definition}
 Let $\sigma$ be a store and $\av$ be an avail.
Let $N$ be a set of nodes.
For $n \in N$, we write $\puse(n)$ for the set of variables used in $n$.
We define a binary relation $\sigma \approx_n \av$ as follows:
\begin{align*}
 \forall x \in \puse(n). \sigma(x) = \av(n,x).
\end{align*}
We simply write $\sigma \approx \av$ for $\forall n \in N. \sigma
 \approx_n \av$.
\end{definition}

Next, we define a notation that represent reachable and unreachable
nodes from some node in a CFG.

\begin{definition}
Let $n \in N$.
We define the following sets:
\begin{align*}
\UR(n) &= \{m ~|~ \text{$m$ is unreachable from $n$ in $P$}\} \\
\UR'(n) &= \{m ~|~ \text{$m$ is unreachable from $n$ in the graph
 obtained by removing} \\
& \qquad \qquad \text{all back edges from $P$} \} \\
R(n) &= \{m ~|~ \text{$m$ is reachable from $n$ in $P$}\} \\
R'(n) &= \{m ~|~ \text{$m$ is reachable from $n$ in the graph
 obtained by removing} \\
& \qquad \qquad \text{all back edges from $P$} \} 
\end{align*}
\end{definition}

\subsection{Proof of Property 1}
The fact that a PDG can execute sentences in the same order as the
corresponding CFG is executed means that any statement that should
be executed next in the CFG is also executable at the current state of
the run of the PDG, that is, it satisfies $\pcondCFLD$ (Lemma
\ref{lemma:sequential_execution}).
Therefore our first target is to prove this fact.

When a statement is to be executed next in a CFG run, a statement
on which it is loop-independently data-dependent or def-order-dependent
should have already been executed, or is contained in the branch of some
conditional statement that are not executed.
The next two lemmata formally describe this fact.
They states that if node $q$ is loop-independently data-dependent or
def-order-dependent on $p$, $p$ is not reachable from $q$ in the graph 
obtained by removing all back edges from the CFG.

\begin{lemma} \label{lemma:LIDD}
 For all $p,q \in N$, if $\LIDD(p,q)$ holds then $p \in \UR(q)$ or $p \in
 \UR'(q) - \UR(q)$ holds.
\end{lemma}
\begin{proof}
Since $p \in \UR(q)$ or $p \in \UR'(q) - \UR(q)$ implies $p \in \UR(q)
 \cup \UR'(q)$ and $\UR(q) \subseteq \UR'(q)$, we just need to show $p
 \in \UR'(q)$.
Suppose that $\LIDD(p,q)$ and $p \not \in \UR'(q)$ hold.
$p \not \in \UR'(q)$ is equivalent to $p \in R'(q)$.
Thus, $p$ is reachable from $q$ without passing any back edge.
Since $\LIDD(p,q)$ holds, by the definition of loop-independent data
 dependence, either i) $q$ is reachable from $p$ and $p$ and $q$ are not
 contained in the same loop, or ii) if $p$ and $q$ are contained in the
 same loop then there is a reaching path from $q$ to $p$ that does not
 contain any back edge.
However, both cases contradict to $p \in R'(q)$.
\end{proof}

\begin{lemma} \label{lemma:Deforder}
 For all $p, q \in N$, if $\DefOrd(p,q)$ holds then $p \in \UR(q)$ or $p
 \in \UR'(q) - \UR(q)$ holds.
\end{lemma}
\begin{proof}
 It is again sufficient to show $p \in \UR'(q)$.
Suppose $DefOrd(p,q)$ and $p \not \in \UR'(q)$.
$p \not \in \UR'(q)$ is equivalent to $p \in R'(q)$.
Thus, $p$ is reachable from $q$ without passing any back edge.
Since $\DefOrd(p,q)$ holds, by the definition of def-order dependence,
 either i) $q$ is reachable from $p$ and $p$ and $q$ are not
 contained in the same loop, ii) $p$ and $q$ are contained in the
 same loop and $q$ is reachable from $p$ without passing any back edge,
 or iii) $p$ and $q$ are contained in the same loop, $p$ is reachable
 from $q$ and vice versa, but only by passing back edges.
However, either case contradicts to $p \in R'(q)$.
\end{proof}

The next lemma states that if a run of a PDG corresponds to a run of
the corresponding CFG, then at the current state of the run of the PDG,
any loop-independent data dependence and def-order dependence from the
statements that are already executed in a run of the CFG are
satisfied (i.e., the states of the edges are $\chk$).

\begin{lemma} \label{lemma:condFD}
 Suppose that there is a run $s_0 \stackrel{n_0}{\to} s_1
 \stackrel{n_1}{\to} \dots \stackrel{n_{k-1}}{\to} s_k$  of $G$ that
 corresponds to a run $\sigma_0 \stackrel{n_0}{\to} \sigma_1
 \stackrel{n_1}{\to} \dots \stackrel{n_{k-1}}{\to} \sigma_k$ of $P$,
 that is to say, $\forall i. \sigma_i \approx_{n_i} \av_i$.
If $\sigma_k \stackrel{n_k}{\to} \sigma_{k+1}$, the following
 holds:
\begin{align*}
\forall p \in \mathit{UR'}(n_k). (p,q) \in F \oplus D \Rightarrow
 \ec_k(p,q) = \chk.
\end{align*}
\end{lemma}

\begin{proof}
 Since $\UR(n) \subseteq \UR'(n)$, the target of the proof can be
 divided into the following two claims.
\begin{align}
\forall p \in \mathit{UR}(n_k). (p,q) \in F \oplus D \Rightarrow
 \ec_k(p,q) = \chk \\
\forall p \in \mathit{UR'}(n_k) - \mathit{UR}(n_k). (p,q) \in F \oplus D
 \Rightarrow \ec_k(p,q) = \chk
\end{align}
We prove this by induction on $k$.
 \\
\noindent \textbf{(Base Case)}
If $k=0$, the claims hold since $\UR(n_0)=\UR'(n_0)=\emptyset$.

\noindent \textbf{(Induction Step)}
\begin{enumerate}
 \item If $n_{k-1} = (x := e)$, $n_k$ has a unique successor.
\begin{enumerate}
 \item If $n_k$ and $n_{k-1}$ are not contained in the same loop, we
       have $\UR(n_k) = \UR(n_{k-1}) \cup \{n_{k-1}\}$.
By the induction hypothesis, we have $\forall p \in
       \UR(n_{k-1}). \forall (p,q) \in F \oplus D. \ec_{k-1}(p,q) =
       \chk$.
By the definition of $\udec$, $\ec_k(n_{k-1},q) = \chk$ and $\forall
       (p,q) \in F \oplus D. p \neq n_{k-1} \Rightarrow \ec_k(p,q) =
       \ec_{k-1}(p,q)$.
Therefore, claim (1) holds.
Since $\UR'(n_k) = \UR(n_k)$, claim (2) also holds.
\item If $n_k$ and $n_{k-1}$ are contained in the same loop,
      $\UR(n_k) = \UR(n_{k-1})$.
Since $n_{k-1} \not \in \UR(n_k)$, by the induction hypothesis and the
      definition of $\udec$, we have $\forall p \in \UR(n_k). \forall
      (p,q) \in F \oplus D \Rightarrow \ec_k(p,q) = \chk$.
That is, claim (1) holds.
Meanwhile, because $\UR'(n_k) = \UR'(n_{k-1}) \cup \{n_{k-1}\}$,
\begin{align*}
\UR'(n_k) - \UR(n_k) &= (\UR'(n_{k-1}) \cup \{n_{k-1}\}) - \UR(n_{k-1})
 \\
&= (\UR'(n_{k-1}) - \UR(n_{k-1})) \cup (\{n_{k-1}\} - \UR(n_{k-1})) \\
&= (\UR'(n_{k-1}) - \UR(n_{k-1})) \cup \{n_{k-1}\} \quad (\because n_{k-1}
 \not \in \UR(n_{k-1}))
\end{align*}
By the induction hypothesis and the definition of $\udec$, we have
      $\forall p \in \UR'(n_{k-1}) - \UR(n_{k-1}). (p,q) \in F \oplus D
      \Rightarrow \ec_k(p,q) = \chk$.
Again by the definition of $\udec$, $\forall (n_{k-1},q) \in F \oplus
      D. \ec_k(n_{k-1},q) = \chk$.
This implies claim (2).
\end{enumerate}
\item If $n_{k-1} = (\iif e)$, assume that $\Se(e)\sigma_{k-1} =
      \Se_{\av}(n_{k-1},e)\av_{k-1}=\BT$.
Suppose $(n_{k-1}, p_t)_\BT \in E$ and $(n_{k-1}, p_f)_\BF \in E$.
Note that $n_k = p_t$.
\begin{enumerate}
 \item If $n_k$ and $n_{k-1}$ are not contained in the same loop,
       $\UR(n_k) = (\UR(n_{k-1}) \cup \{n_{k-1}\}) \cup (R(p_f) -
       R(p_t))$.
Because $\UR'(n_k) = \UR(n_k)$, claim (2) immediately holds.
We show claim (1).
For a statement $p \in \UR(n_{k-1})$, by the induction hypothesis we
       have $(p,q) \in F \oplus D \Rightarrow \ec_{k-1}(p,q) = \chk$.
Since $n_{k-1}$ is an $\iif$ statement, there are no $F$- and $D$-edges
       from $n_{k-1}$.
For a statement $p \in R(p_f) - R(p_t)$, if $p \in G_F(n_{k-1}) -
       G_T(n_{k-1})$ then by the definition of $\udec$ we have $\forall
       p \in G_F(n_{k-1}) - G_T(n_{k-1}). (p,q) \in F \oplus D
       \Rightarrow \ec_k(p,q) = \chk$.
If not, that is to say, $p$ is reachable from $n_{k-1}$ only via a
       looping edge (i.e., in the CDG $G$), we have $(r,p) \in C -
       \widehat{C}$ for some $r \in G_F(n_{k-1}) - G_T(n_{k-1})$.
Then, by the definition of looping edges, there is $o \neq r$ such that
       $(o,p) \in C$ and $\exists Q \in \{\BT,\BF\}. n_{k-1}
       \in G_Q(o)$ and $p \in G_Q(o)$.
By Lemma \ref{lemma:looping_edge_back_edge}, we have $p \in
       \UR'(n_{k-1}) - \UR(n_{k-1})$ and we apply the induction
       hypothesis.
\item If $n_k$ and $n_{k-1}$ are contained in the same loop, $\UR(n_k) =
      \UR(n_{k-1})$.
Therefore, claim (1) holds as in case 1-(b).
Meanwhile, we have $\UR'(n_k) = \UR'(n_{k-1}) \cup \{n_{k-1}\} \cup
      (R'(p_f) - R'(p_t))$.
Therefore, $\UR'(n_k) - \UR(n_k) = (\UR'(n_{k-1}) - \UR(n_k)) \cup
      (\{n_{k-1}\} - \UR(n_k)) \cup ((R'(p_f) - R'(p_t)) - \UR(n_k))$.
By the induction hypothesis and the definition of $\udec$, $\forall p
      \in \UR'(n_{k-1}) - \UR(n_k)). (p,q) \in F \oplus D \Rightarrow
      \ec_k(p,q) = \chk$.
Furthermore, since $\{n_{k-1}\} - \UR(n_{k-1}) = \{ n_{k-1}\}$ and
      $n_{k-1}$ is an $\iif$ statement, there are no $F$- and $D$-edges
      from $n_{k-1}$.
Finally, since $(R'(p_f) - R'(p_t)) - \UR(n_k) = R'(p_f) - R'(p_t)$, for
      $p \in R'(p_f) - R'(p_t)$ it is trivial that $p \in G_F(n_{k-1}) -
      G_T(n_{k-1})$. Thus, by the definition of $\udec$, claim (2)
      holds.
\end{enumerate}
\end{enumerate}
\end{proof}

The above lemma will be used to prove that when a PDG is executed in
the same order as the corresponding CFG is executed, the statements that
should be executed next in the CFG satisfy $\pcondF$ and $\pcondD$ at
the current state of the run of the PDG.

\begin{lemma} \label{lemma:looping_edge_back_edge}
 Suppose that $p$ is an $\iif$ statement and has outgoing looping
 edges whose label is $Q$.
We assume $\exists r \neq p. \exists p' \in G_{Q'}(r). \{p,p'\}\subseteq
 G_{Q'}(r)$.
If node $q$ is reachable from $p$ only via a looping edge outgoing from
 $p$, then $q$ is reachable from $p$ only via a back edge in the CFG
 $P$.
\end{lemma}
\begin{proof}
Since $Q$-edges from $p$ in $G$ are looping edges, $G_Q(p)$ consists of
 a loop in $P$.
Therefore, there is $q'$ in the loop that has an incoming edge from the
 outside of the loop.
The reason why such $q'$ exists is because there is $r \neq p$ such that
 $q' \in G_{Q'}(r)$.
By the definition of the back edge, the edge $(p,q')_Q \in E$ of $P$ is a
 back edge of the loop.
Now assume that a node $q$ in the loop is reachable from $p$ only via a
 looping edge outgoing from $p$.
That means $q$ is reachable from $q'$ by the definition of control
 dependence.
Therefore, $q$ is reachable from $p$ only via the back edge $(p,q')_Q$.
\end{proof}

The next lemma states that if there exists a run of $G$ corresponding to
the execution of $P$, the states of the $C$-edges outgoing from the
statements in the subgraph of the statement that will be executed next
in $P$ are not $\act$.

\begin{lemma} \label{lemma:condC1}
 Suppose that there is a run $s_0 \stackrel{n_0}{\to} s_1
 \stackrel{n_1}{\to} \dots \stackrel{n_{k-1}}{\to} s_k$  of $G$ that
 corresponds to a run $\sigma_0 \stackrel{n_0}{\to} \sigma_1
 \stackrel{n_1}{\to} \dots \stackrel{n_{k-1}}{\to} \sigma_k$ of $P$,
 that is to say, $\forall i. \sigma_i \approx_{n_i} \av_i$.
If $\sigma_k \stackrel{n_k}{\to} \sigma_{k+1}$, then $\forall p \in
 G^*(n_k). p \neq n_k \Rightarrow \forall c(p,q) \in C. ec_k(c(p,q))
 \neq \act$.
\end{lemma}
\begin{proof}
 By induction on $k$.\\
\noindent \textbf{(Base Case)}
$k=0$.
By the definition of the operational semantics of the PDG, we have $\forall
 (p,q) \in C. \ec_0(p,q) = \act \Rightarrow p = \pentry$.
If $n_0$ is not an $\iif$ statement, the claim immediately follows.
If $n_0$ is an $\iif$ statement, by the definition of the augmented CFG,
it is not the case that $(\pentry, q) \in C$ for some $q \in G^*(n_0)$.
Thus, the claim follows.

\noindent \textbf{(Induction Step)}
Since the claim trivially follows if $n_k$ is not an $\iif$ statement,
 we only consider the case where $n_k$ is an $\iif$ statement.
For some $i<k$, we consider the case where there is $n_i$ such that
 $\CD(n_i,n_k) \wedge \forall j \in (i,k). \neg \CD(n_j,n_k)$.
For $p \in G^*(n_k)$, since $G^*(n_k) \subseteq G^*(n_i)$, by the
 induction hypothesis, $\ec_i(c(p,q)) \neq \act$.
Since node $n_j$ for any $j \in (i,k)$ does not have a $C$-edge pointing
 to a node in $G^*(n_k)$, we have $\ec_k(c(p,q))=\ec_i(c(p,q)) \neq
 \chk$ by the definition of $\udec$.
\end{proof}

The next lemma states that if there exists a run of $G$ corresponding to
the execution of $P$, the loop-carried data dependence from the
statement that will be executed next in $P$ is satisfied.

\begin{lemma} \label{lemma:condL1}
Suppose that there is a run $s_0 \stackrel{n_0}{\to} s_1
 \stackrel{n_1}{\to} \dots \stackrel{n_{k-1}}{\to} s_k$  of $G$ that
 corresponds to a run $\sigma_0 \stackrel{n_0}{\to} \sigma_1
 \stackrel{n_1}{\to} \dots \stackrel{n_{k-1}}{\to} \sigma_k$ of $P$,
 that is to say, $\forall i. \sigma_i \approx_{n_i} \av_i$.
If $\sigma_k \stackrel{n_k}{\to} \sigma_{k+1}$, then $\forall l(n_k,q)
 \in L. \ec_k(l(n_k,q))=\chk$.
\end{lemma}
\begin{proof}
 If $(n_k,q) \in L$, by the definition of the loop-carried data
 dependence, $n_k$ and $q$ reside in the same loop and every reaching
 path from $n_k$ to $q$ passes a back edge of the loop.
Suppose that $p \in \mca(n_k,q)$ be the last executed minimal common
 ancestor of $n_k$ and $q$.
If $q$ was executed during the execution from $p$ to $n_k$, the state of
 $l(n_k,q)$ turned to be $\chk$ according to $\udec$.
Let $o$ be an arbitrary $\iif$ statement executed after $q$'s execution.
Then we have $\forall Q \in \{\BT,\BF\}. \{q,n_k\}\not \in G_Q(o)$.
Otherwise, $o$ should have been the last executed minimal common
 ancestor of $q$ and $n_k$.
Therefore, if $n_k \not \in G(o)$, the state of $l(n_k,o)$ will not be
 changed by the execution of $o$.
Meanwhile, if we assume $n_k \in G_T(o)$ (the case of $G_F(o)$ is the
 same), since $n_k$ will be executed, the condition of $o$ was evaluated
 to $\BT$.
Thus, by the definition of $\udec$, the state of $l(n_k,q)$ will not be
 changed.
As a consequence, we have $\ec_k(l(n_k,q))=\chk$.
If $q$ is not executed during the execution from $p$ to $n_k$, there is
 an $\iif$ statement $o$ such that $q \in G(o)$ (note that $o$ may
 coincide with $p$).
If we assume $q \in G_T(o)$ (the case of $G_F(o)$ is the same), since
 $q$ was not executed, the condition of $o$ had been
 evaluated to $\BF$ and $q \in G_F(o) - G_T(o)$.
By the same argument as above, we have $\forall Q \in
 \{\BT,\BF\}. \{q,n_k\}\not \in G_Q(o)$.
By the definition of $\udec$, the state of $l(n_k,q)$ turned to be $\chk$
 after $o$ was executed.
As for the $\iif$ statements that will be executed after that, the same
 argument applies as above and the state of $l(n_k,q)$ will be
 unchanged.
Thus, we have $\ec_k(l(n_k,q))=\chk$.
\end{proof}

This lemma will be used to show that the statement to be executed next
in a run of the CFG satisfies $\pcondL$ in the current state of the
run of the PDG.

By using the lemmata \ref{lemma:condFD}, \ref{lemma:condC1} and
\ref{lemma:condL1}, we can show that there is a run of $G$ that has the
same execution order as that of $P$.

\begin{lemma} \label{lemma:sequential_execution}
Suppose that there is a run $s_0 \stackrel{n_0}{\to} s_1
 \stackrel{n_1}{\to} \dots \stackrel{n_{k-1}}{\to} s_k$  of $G$ that
 corresponds to a run $\sigma_0 \stackrel{n_0}{\to} \sigma_1
 \stackrel{n_1}{\to} \dots \stackrel{n_{k-1}}{\to} \sigma_k$ of $P$,
 that is to say, $\forall i. \sigma_i \approx_{n_i} \av_i$.
If $\sigma_k \stackrel{n_k}{\to} \sigma_{k+1}$, then $n_k \in
 \Next(s_k)$.
\end{lemma}
\begin{proof}
 We show that $\pcondCFLD(s_k,n_k)$ is the case.
\begin{description}
 \item[Proof of $\pcondC(s_k,n_k)$.] 
The predicate $\pcondC(s_k,n_k)$ is divided as the following two
	    conjuncts:
\begin{align}
&\exists c(p,n_k) \in C. \ec(c(p,n_k)) = \act \\
&\forall q \in G^*(n_k).q \neq n_k \Rightarrow \forall c(r,q) \in
 C. \ec(c(r,q)) \neq \act.
\end{align}
In the following, we prove them individually.
\begin{itemize}
 \item If $n_{k-1}$ is not an $\iif$ statement, since $n_{k-1} \in
       \Next(s_{k-1})$, there is an $\iif$ statement $q$ such that
       $\ec_{k-1}(c(q,n_{k-1})) = \act$.
Since $n_k$ is the unique successor of $n_{k-1}$, $\exists Q \in
      \{\BT, \BF\}. \{n_{k-1}, n_k\} \subseteq G_Q(q)$.
Therefore, we have $\ec_{k-1}(c(q,n_k)) = \act$.
Since the execution of $n_{k-1}$ does not change the state of $c(q,n_k)$
       by the definition of $\udec$, we have $\ec_k(c(q,n_k)) = \act$.
If $n_{k-1}$ is an $\iif$ statement\footnote{Here we assume that
       $n_{k-1}$ has successors $n_k$ and the one different from
       $n_k$. Otherwise the above argument applies.}, $c(n_{k-1}, n_k)
       \in C$.
Since $\sigma_k \stackrel{n_k}{\to} \sigma_{k+1}$ and $\sigma_k(e) =
       \Se(n_k,e)\av_k$, we have $\ec_k(c(n_k,p)) = \act$ by the
       definition of $\udec$.
Thus, conjunct (3) follows.
\item By Lemma \ref{lemma:condC1}, conjunct (4) follows.
\end{itemize}
\item[Proof of $\pcondF(s_k,n_k)$.]
For all $f(p,n_k) \in F$, by Lemma \ref{lemma:LIDD}, $p \in
	   \UR(n_k)$ or $p \in \UR'(n_k) - \UR(n_k)$ holds.
Then, by Lemma \ref{lemma:condFD}, we have $\ec_k(f(p,n_k)) = \chk$ and
	   the claim follows.
\item[Proof of $pcondD(s_k,n_k)$.]
For all $d(p,n_k) \in F$, by Lemma \ref{lemma:Deforder}, $p \in
	   \UR(n_k)$ or $p \in \UR'(n_k) - \UR(n_k)$ holds.
Then, by Lemma \ref{lemma:condFD}, we have $\ec_k(f(p,n_k)) = \chk$ and
	   the claim follows.
\item[Proof of $\pcondL(s_k,n_k)$.]
It immediately follows from Lemma \ref{lemma:condL1}.
\end{description}
\end{proof}

By this lemma, the following important theorem is proved.
This theorem states that for any run of $P$ there is a corresponding run
of $G$.

\begin{theorem} \label{thm:corresponding_run}
Suppose $\sigma_0 \approx \av_0$.
For a run $\sigma_0 \stackrel{n_0}{\to} \sigma_1 \stackrel{n_1}{\to}
 \dots \stackrel{n_k}{\to} \sigma_{k+1}$ of $P$, there is a run $s_0
 \stackrel{n_0}{\to} s_1 \stackrel{n_1}{\to} \dots \stackrel{n_k}{\to}
 s_{k+1}$ of $G$ that satisfies $\forall i \in [0,k]. \sigma_i
 \approx_{n_i} \av_i$.
\end{theorem}
\begin{proof}
By induction on $k$ .\\
\noindent \textbf{(Base Case)}
$k=0$.
For a run $\sigma_0 \stackrel{n_0}{\to} \sigma_1$ of $P$, since we know
 $\sigma_0 \approx \av_0$ by the assumption, $n_0 \in \Next(s_0)$
 follows by Lemma \ref{lemma:sequential_execution}.
Henceforth, there is a run $s_0 \stackrel{n_0}{\to} s_1$ of $G$.
Moreover, by the assumption we have $\sigma_0 \approx_{n_0} \av_0$.

\noindent \textbf{(Induction Step)}
By the induction hypothesis, for $P$'s run $\sigma_0 \stackrel{n_0}{\to}
 \sigma_1 \stackrel{n_1}{\to} \dots \stackrel{n_k}{\to} \sigma_{k+1}$,
 there is a $G$'s run $s_0 \stackrel{n_0}{\to} s_1 \stackrel{n_1}{\to}
 \dots \stackrel{n_k}{\to} s_{k+1}$ that satisfies $\forall i \in
 [0,k]. \sigma_i \approx_{n_i} \av_i$.
Suppose that $\sigma_{k+1} \stackrel{n_{k+1}}{\to} \sigma_{k+2}$.
By Lemma \ref{lemma:sequential_execution}, we have $n_{k+1} \in
 \Next(s_{k+1})$.
Thus, there is a $G$'s run $s_0 \stackrel{n_0}{\to} s_1
 \stackrel{n_1}{\to} \dots \stackrel{n_k}{\to} s_{k+1}
 \stackrel{n_{k+1}}{\to} s_{k+2}$.
We then show $\sigma_{k+1} \approx_{n_{k+1}} \av_{k+1}$, that is to say,
 $\forall x \in \puse(n_{k+1}). \sigma_{k+1} = \av_{k+1}(n_{k+1}, x)$.
For $x \in \puse(n_{k+1})$, if an assignment statement defining $x$ is
 executed between $n_0$ and $n_k$, then let the last one among such
 statements be $n_j = (x := e)$.
Then we have $(n_j, n_{k+1}) \in F \oplus L$.
By the definition of $\udav$, $\av_{k+1}(n_{k+1}, x) = \Se(n_j,
 e)\av_j$ holds.
Since we have $\sigma_j \approx_{n_j} \av_j$ by the induction hypothesis,
 $\Se(e)\sigma_j = \Se(n_j,e)\av_j$ follows.
Meanwhile, by the operational semantics of the CFG, we have
 $\sigma_{j+1}(x) = \Se(e)\sigma_j$.
Hence, $\sigma_{j+1}(x) = \av_{j+1}(n_{k+1}, x)$.
Since no assignment statements defining $x$ are executed between $n_j$ and
 $n_{k-1}$, we have $\sigma_{j+1}(x) = \sigma_{k+1}(x)$ and
 $\av_{j+1}(n_{k+1},x) = \av_{k+1}(n_{k+1},x)$.
Therefore, $\sigma_{k+1}(x) = \av_{k+1}(n_{k+1},x)$ follows.
If no assignment statements defining $x$ are executed between $n_0$ and
 $n_k$, then we have $\sigma_{k+1}(x) = \sigma_0(x)$ and
 $\av_{k+1}(n_{k+1},x) = \av_0(n_{k+1},x)$.
Since $\sigma_0 \approx \av_0$ from the assumption, $\sigma_{k+1}(x) =
 \av_{k+1}(n_{k+1},x)$ follows.
\end{proof}

\subsection{Proof of Property 2}
Next we prove Property 2.
Property 2 says that all runs of a dPDG starting with the same initial
state end with the same final state (if they terminate).
This property is derived from \emph{confluence} of the runs of a dPDG
(Lemma \ref{lemma:diamond_property}) which says if there are multiple
executable statements in a run of a dPDG, the order of execution does
not change the final result.
This property is proved by the fact that an execution of a statement
does not affect executability and the execution results of the other
statements that are simultaneously executable in the dPDG (Lemma
\ref{lemma:NextFL}, \ref{lemma:NextD}, \ref{lemma:dPDGcondC1} and
\ref{lemma:dPDG_non_interference}).
Other lemmata are used to prove the above lemmata or Theorem
\ref{thm:dpdg}.

The next lemma states that two executable nodes at a run of a PDG are
contained in the subgraph of the same truth value of some $\iif$
statement.

\begin{lemma} \label{lemma:executable_then_mca}
 Let $s_0 \stackrel{n_0}{\to} s_1 \stackrel{n_1}{\to} \dots$ be a run of
 a PDG.
Suppose $c(o_1,p) \in C$, $c(o_2,q) \in C$, $p \not \in G^*(q)$ and $q
 \not \in G^*(p)$.
For any $i$, if $\ec_i(o_1,p) = \ec_i(o_2,q) = \act$ then $\exists
 r \in \mca(p,q). \exists Q \in \{\BT,\BF\}. \{p,q\} \subseteq G^*_Q(r)$.
\end{lemma}
\begin{proof}
 By induction on $i$.\\
\noindent \textbf{(Base Case)}
If $i=0$, $o_1=o_2=\pentry$.
Obviously, $\{p,q\} \subseteq G_T(\pentry)$.
Since $\mca(p,q) = \{\pentry\}$, the claim holds.

\noindent \textbf{(Induction Step)}
\begin{enumerate}
 \item If $\ec_{i-1}(o_1,p) = \ec_{i-1}(o_2,q) = \act$, by the induction
       hypothesis the claim holds.
 \item If $\ec_{i-1}(o_1,p) \neq \act$ and $\ec_{i-1}(o_2,p) \neq \act$,
       then by the definition of $\udec$ and $\pcondC$, $o_1 = o_2 =
       n_{i-1}$ and $\exists Q \in \{\BT,\BF\}. \{p,q\} \subseteq
       G_Q^*(n_{i-1})$ must be the case.
       Here suppose $n_{i-1} \not \in \mca(p,q)$, then in the CDG either
       $p$ or $q$ must lie on a path from $\pentry$ to $n_{i-1}$.
       This contradicts to either $q \in G^*(p)$ or $p \in G^*(q)$.
       Henceforth, $n_{i-1} \in \mca(p,q)$.
 \item If $\ec_{i-1}(o_1,p) = \act$ and $\ec_{i-1}(o_2,p) \neq \act$,
       then $n_{i-1} = o_2$ and must be some $(o,o_2) \in C$ such that $\ec_{i-1}(o,o_2) =
       \act$ by the definition of $\udec$ and $\pcondC$.
       By the induction hypothesis, there is $r_1 \in \mca(p,o_2)$ and
       $Q \in \{\BT, \BF\}$ such that $\{p,o_2\} \subseteq G^*_Q(r_1)$.
       Since $(o_2,q) \in C$, we have $\{p,q\} \subseteq G^*_Q(r_1)$.
       If we assume that $r_1 \not \in \mca(p,q)$, either $p$ or $q$
       must lie on a path from $\pentry$ to $r_1$ in the CDG.
       This contradicts to either $q \in G^*(p)$ or $p \in G^*(q)$.
       Henceforth, $n_{i-1} \in \mca(p,q)$.
\end{enumerate}
\end{proof}

\begin{lemma} \label{lemma:NextC}
 Let $s_0 \stackrel{n_0}{\to} s_1 \stackrel{n_1}{\to} \dots$ be a run of
 a PDG.
For any $p$ and $q$ ($p \neq q$), if $\{p,q\} \subseteq
 \Next(s_i)$, then $p \not \in G^*(q)$ and $q \not \in G^*(p)$.
\end{lemma}
\begin{proof}
 $p \in G^*(q)$ contradicts to the second conjunct of the predicate $\pcondC$.
$q \in G^*(p)$ does the same.
\end{proof}

\begin{lemma} \label{lemma:NextFL}
 Let $s_0 \stackrel{n_0}{\to} s_1 \stackrel{n_1}{\to} \dots$ be a run of
 a PDG.
For any $p$ and $q$ ($p \neq q$), if $\{p,q\} \subseteq \Next(s_i)$ then
 there are no $F$- and $L$-edges between $p$ and $q$.
\end{lemma}
\begin{proof} We split the proof into the following cases.
 \begin{itemize}
  \item Let $(p,q) \in F$.
	Since $p \in \Next(s_i)$, there is $(o,p) \in C$ such that
	$\ec_i(o,p) = \act$.
	Therefore, there is some $j < i$ such that $n_j = o$ and
	$\forall k \in (j,i). \ec_k(o,p) = \act$ (this implies $p$ is
	not executed between $s_j$ and $s_i$).
	By the definition of $\udec$, $\ec_{j+1}(f(p,q)) = \unchk$.
	Furthermore, by the definition of $\pcondC$, we have $\forall k
	\in (j,i). p \not \in G(n_k)$ since $\ec_k(o,p) = \act$.
	Therefore, $\ec_i(f(p,q))=\ec_{j+1}(f(p,q)) = \unchk$.
	This contradicts to $q \in Next(s_i)$.
  \item Let $(p,q) \in L$.
	Since $q \in \Next(s_i)$, there is $(o,q) \in C$ such that
	$\ec_i(o,q) = \act$.
	Therefore, there is some $j < i$ such that $n_j = o$ and
	$\forall k \in (j,i). \ec_k(o,q) = \act$ (this implies $q$ is
	not executed between $s_j$ and $s_i$).
	By the definition of $\udec$, $\ec_{j+1}(l(p,q)) = \unchk$.
	Furthermore, by the definition of $\pcondC$, we have $\forall k
	\in (j,i). q \not \in G(n_k)$ since $\ec_k(o,q) = \act$.
	Therefore, $\ec_i(l(p,q))=\ec_{j+1}(l(p,q)) = \unchk$.
	This contradicts to $p \in Next(s_i)$.
  \item The proofs of the cases for $(q,p) \in F$ and $(q,p) \in L$ are
	the same as above.
 \end{itemize}
\end{proof}

\begin{lemma} \label{lemma:NextD}
Let $s_0 \stackrel{n_0}{\to} s_1 \stackrel{n_1}{\to} \dots$ be a run of
 a deterministic PDG.
For any $p$ and $q$ ($p \neq q$), if $f_x(p,n) \in F$ and $f_x(q,n) \in
 F$ for some $n$, then $\{p,q\} \not \subseteq \Next(s_i)$ holds.
\end{lemma}
\begin{proof}
 Suppose $\{p,q\} \subseteq \Next(s_i)$.
Since $p$ and $q$ are assignment statements, we know $p \not \in G^*(q)$
 and $q \not \in G^*(p)$.
By Lemma \ref{lemma:executable_then_mca}, there must be some $r \in
 \mca(p,q)$ and $Q \in \{\BT,\BF\}$ such that $\{p,q\} \subseteq G^*_Q(r)$.
Then, by the condition of the dPDG, either $d(p,q) \in D$ or $d(q,p) \in D$
 holds.
We consider the case $d(p,q) \in D$.
Since $p \in \Next(s_i)$, there is $(o,p) \in C$ such that $\ec_i(o,p) =
 \act$.
Therefore, there is some $j < i$ such that $n_j = o$ and $\forall k \in
 (j,i). \ec_k(o,p) = \act$ (this implies $p$ is not executed between
 $s_j$ and $s_i$).
By the definition of $\udec$, $\ec_{j+1}(d(p,q)) = \unchk$.
Furthermore, by the definition of $\pcondC$, we have $\forall k
\in (j,i). p \not \in G(n_k)$ since $\ec_k(o,p) = \act$.
Therefore, $\ec_i(d(p,q))=\ec_{j+1}(d(p,q)) = \unchk$.
This contradicts to $q \in Next(s_i)$.
\end{proof}

\begin{lemma} \label{lemma:mca_subsumption}
 If $o \in G^*(p)$ and $p \not \in G^*(q)$ and $q \not \in G^*(p)$, then 
 $\mca(p,q) \subseteq \mca(o,q)$ holds.
\end{lemma}
\begin{proof}
Since CDGs are connected and $\pentry$ is the root, there is a
 non-cyclic path $\pi=\pentry \dots r$ in the CDG.
Suppose $r \in \mca(p,q)$.
By the definition of the mca, we have two non-cyclic paths $\pi_1=\pi
 \dots p$ and $\pi_2=\pi \dots q$.
Since $o \in G^*(p)$, there is a non-cyclic path $\pi_3=\pi_1 \dots o$.
Here since the longest common prefix of $\pi_1$ and $\pi_2$ is $\pi$, and
 $q \not \in G^*(p)$ and $p \not \in G^*(q)$, we have $r \neq p$
 and $r \neq q$.
Therefore, the longest common prefix of $\pi_2$ and $\pi_3$ is $\pi$.
This shows $r \in \mca(o,q)$.
\end{proof}

The next lemma states that for a run of a dPDG, if there are two
executable statements, the execution of one node does not
affect the condition $\pcondC$ of the other.

\begin{lemma} \label{lemma:dPDGcondC1}
 Let $s_0 \stackrel{n_0}{\to} s_1 \stackrel{n_1}{\to} \dots$ be a run of
 a dPDG.
For any $p$ and $q (p \neq q)$, we assume $\{p,q\} \subseteq
 \Next(s_i)$.
If $s_i \stackrel{q}{\to} s_{i+1}$, the following holds:
\begin{equation}
\forall r \in G^*(p). r \neq p \Rightarrow \forall (w,r) \in
 C. \ec_{i+1}(w,r) \neq \act.
\end{equation}
\end{lemma}
\begin{proof}
 Suppose that claim (5) is violated by the execution of $q$.
By the definition of $\udec$, there exists some $r \in G^*(p)$ such that
 $(q,r) \in C$.
Since $\{p,q\} \subseteq \Next(s_i)$, by Lemma \ref{lemma:NextC}, we
 have $p \not \in G^*(q)$ and $q \not \in G^*(p)$.
By Lemma \ref{lemma:executable_then_mca}, there exists some $u \in
 \mca(p,q)$ and $Q \in \{\BT,\BF\}$ such that $\{p,q\} \subseteq
 G_Q^*(u)$.
Meanwhile, since $r \in G^*(p)$ and $(q,r) \in C$, there exists some $o
 \in G^*(p) - G^*(q)$ such that $r \in G^*(o)$.
That is to say, we have $r \in G^*(o)$ and $r \in G^*(q)$.
This implies $\forall v \in \mca(o,q). \forall Q \in
 \{\BT,\BF\}. \{o,q\} \not \subseteq G_Q^*(v)$ from the first condition
 of the dPDG.
However, by Lemma \ref{lemma:mca_subsumption}, we have $u \in \mca(p,q)
 \subseteq \mca(o,q)$ which leads to contradiction.
\end{proof}

\begin{lemma} \label{lemma:outgoing_looping_edge}
 Let $s_0 \stackrel{n_0}{\to} s_1 \stackrel{n_1}{\to} \dots$ be a run of
 a dPDG and $q \in \Next(s_i)$.
For $o \in G(q)$ such that $(o,p) \in F \oplus D$, we assume $\ec_i(o,p)
 = \chk$.
Then $q$ has an outgoing looping edge.
\end{lemma}
\begin{proof}
 We assume that all $C$-edges from $q$ are not looping edges.
Since $q \in \Next(s_i)$, there is $(r,q) \in C$ such that
$\ec_i(r,q) = \act$.
Therefore, there is some $j < i$ such that $n_j = r$ and
$\forall k \in (j,i). \ec_k(r,q) = \act$.
Suppose $q \in G_T(r)$ (the case of $G_F(r)$ is the same)．
Since $C$-edges from $q$ are not looping edges, we have $G(q) \subseteq
 G_T(r)$.
Therefore, $\forall o \in G(q). \forall (o,p) \in F \oplus
 D. \ec_{j+1}(o,p) = \unchk$ by the definition of $\udec$.
Since $r \in \Next(s_j)$, there are no $C$-edges pointing to a node in
 $G(r)$ whose states are $\act$ at state $s_j$.
Moreover, since $G(q) \subseteq G_T(r)$, there are no $C$-edges pointing
 to a node in $G(q)$ whose states are $\act$ at state $s_{j+1}$.
Thus, $q \in \Next(s_{j+1})$.
If $s_{j+1} \stackrel{v}{\to} s_{j+2}$ (i.e., $\{q, v\} \subseteq
 \Next(s_{j+1})$), by applying Lemma \ref{lemma:dPDGcondC1} we have
 $\forall o \in G^*(q). o \neq q \Rightarrow \forall (w,o) \in
 C. \ec_{j+2}(w,o) \neq \act$.
By repeating this argument, we have $\forall o \in G^*(q). \forall (w,o)
 \in C. \ec_k(w,o) \neq \act$ for any $k \in (j,i]$.
That is to say, no $o \in G(p)$ is executed between $s_j$ and
 $s_i$.
Therefore, $\forall o \in G(q). \forall (o,p) \in F \oplus
 D. \ec_i(o,p) = \ec_{j+1}(o,p) = \unchk$, which leads to
 contradiction.
\end{proof}

The next lemma states that if there are multiple executable statements
in a run of a dPDG, execution of one statement does not interfere with
executability of the other statements.

\begin{lemma} \label{lemma:dPDG_non_interference}
 Let $s$ be a state in a run of a dPDG.
For all $p, q \in \Next(s)$ such that $p \neq q$, if $s
 \stackrel{q}{\to} s'$, then $p \in \Next(s')$ holds.
\end{lemma}
\begin{proof}
We show $\pcondCFLD(s',p)$ holds.
We divide the proof by the type of $q$.
 \begin{enumerate}
  \item If $q$ is an assignment statement, the change of the econf
	caused by the execution of $q$ does not affect the value of the
	predicates $\pcondC$, $\pcondF$ and $\pcondD$.
	Moreover, by Lemma \ref{lemma:NextFL}, $\{p,q\} \subseteq
	\Next(s)$ implies that there are no $L$-edges between $p$ and
	$q$.
	Therefore we have $\pcondCFLD(s',p)$ and $p \in \Next(s')$.
  \item Suppose $q$ is an $\iif$ statement ($\iif e$).
	We prove each condition individually.
\begin{description}
 \item[Proof of $\pcondC(s',p)$.] Since $\pcondC(s,p)$ holds, there is
	    $(o,p) \in C$ that satisfies $\ec(o,p)=\act$.
	    If the execution of $q$ brings $\ec'(o,p) \neq \act$, then
	    $o \in G(q)$ by the definition of $\udec$.
	    This, however, contradicts to $\pcondC(s,q)$ since it
	    implies $p \in G^*(q)$.
	    Therefore, the first conjunct of $\pcondC(s',p)$ holds.
	    The second conjunct follows from Lemma
	    \ref{lemma:dPDGcondC1}.
 \item[Proof of $\pcondF(s',p)$.] Assume that the execution of $q$ causes
	    violation of $\pcondF(s',p)$.
	    We consider the case $\Se(q,e)\av = \BT$ (The case of
	    $\BF$ is the same).
	    By the definition of $\udec$, there must be some $r \in
	    G_T(q) - \{q\}$ such that $(r,p) \in F$ and $\ec(r,p)=
	    \chk$\footnote{Note that $q$'s execution brings $\ec'(r,p)=\unchk$,
	    thus $p$ becomes unexecutable at state $s'$.}.
	    Thus, by Lemma \ref{lemma:outgoing_looping_edge}, $q$ has
	    outgoing looping edges.
	    That is, $q$ is contained in some loop in the CDG by the
	    definition of looping edges.
	    By the third condition of the dPDG, then, $p \in G^*(q)$
	    is the case, which contradicts to $\pcondC(s,q)$.
 \item[Proof of $\pcondD(s',p)$.] The same as the proof of
	     $\pcondF(s',p)$.
 \item[Proof of $\pcondL(s',p)$.] Assume that the execution of $q$ causes
	    violation of $\pcondL(s',p)$.
	    That is, for some $l(p,r) \in L$, $\ec'(l(p,r)) = \unchk$.
	    We consider the case $\Se(q,e)\av = \BT$ (The case of
	    $\BF$ is the same).
	    By the definition of $\udec$, we have $r \in G_T(q) - \{q\}$
	    and $p \in G_T(q)$.
	    This contradicts to $\pcondC(s,q)$.
\end{description}
 \end{enumerate}
\end{proof}

\begin{lemma} \label{lemma:subgraphs_are_disjoint}
 Let $s$ be a state of a run of a dPDG.
If $\{p,q\} \subseteq \Next(s)$, then $G(p) \cap G(q) = \emptyset$
 holds.
\end{lemma}
\begin{proof}
 Assume $G(p) \cap G(q) \neq \emptyset$.
Since $\{p,q\} \subseteq \Next(s)$, we have $p \not \in G^*(q)$ and $q
 \not \in G^*(p)$ by Lemma \ref{lemma:NextC}.
Thus, for some $o_1 \in G(p) - G(q)$ and $o_2 \in G(q) - G(p)$, there is
 $o \in G(p) \cap G(q)$ such that $(o_1,o) \in C$ and $(o_2,o) \in C$.
Since $\{p,q\} \subseteq \Next(s)$, there is $r \in \mca(p,q)$ and $Q \in
 \{\BT,\BF\}$ such that $\{p,q\} \subseteq G_Q^*(r)$ by Lemma
 \ref{lemma:executable_then_mca}.
Since $o_1 \in G(p) - G(q)$ and $o_2 \in G(q) - G(p)$, we have
$o_1 \not \in G^*(q)$ and $o_2 \not \in G^*(p)$.
By applying Lemma \ref{lemma:mca_subsumption} twice, we have $r \in
 \mca(o_1,o_2)$ which implies $\{o_1, o_2\} \subseteq G_Q^*(r)$.
This contradicts to the first condition of the dPDG.
\end{proof}

The next lemma states that if there are two executable statements, the
final result of the execution of them are independent of the execution
order.

\begin{lemma} \label{lemma:diamond_property}
 Let $s$ be a state of a run of a dPDG.
If $\{p,q\} \subseteq \Next(s)$, the following holds:
\[
s \stackrel{p}{\to} s_1 \stackrel{q}{\to} s_2 \Leftrightarrow
s \stackrel{q}{\to} s_1' \stackrel{p}{\to} s_2
\]
\end{lemma}
\begin{proof}
 Suppose $s \stackrel{p}{\to} s_1 \stackrel{q}{\to} s_2$ and $s
 \stackrel{q}{\to} s_1' \stackrel{p}{\to} s_2'$.
We split the proof by the type of $p$ and $q$.
\begin{enumerate}
 \item Both $p$ and $q$ are respectively assignment statements $x:=e_1$
       and $y:=e_2$.
       Since the execution order of $p$ and $q$ do not change the final
       result of the econf, we have $\ec_2 = \ec_2'$.
       If $x \neq y$, we can easily see $\av_2 = \av_2'$ by the
       definition of $\udav$.
       We consider the case $x=y$.
       If there exists $u$ such that $f_x(p,u) \in F$ and $f_x(q,u) \in
       F$, it should be $\{p,q\} \not \subseteq \Next(s)$ by Lemma
       \ref{lemma:NextD} and leads to contradiction.
       Therefore, $\{n ~|~ (p,n) \in F \oplus L\} \cap \{n ~|~ (q,n) \in F \oplus
       L\} = \emptyset$.
       By Lemma \ref{lemma:NextFL} we have $(p,q) \not \in F \oplus L$ and
       $(q,p) \not \in F \oplus L$, which yields $\av_2=\av_2'$.
 \item $p$ is an assignment statement $x:=e_1$ and $q$ is an $\iif$
       statement $\iif e_2$.
       By Lemma \ref{lemma:NextC}, we have $p \not \in G^*(q)$, and thus
       $\ec_2=\ec_2'$.
       Furthermore, we have the following:
\begin{align*}
\av_1 &= \av[(n,x) \mapsto \Se(p,e_1)\av : (p,n) \in F \oplus L] \\
&= \av_2, \\
\av_2' &= \av_1'[(n,x) \mapsto \Se(p,e_1)\av : (p,n) \in F \oplus L].
\end{align*}
Since $\av_1' = \av$, and $(p,q) \not \in F \oplus L$ holds by Lemma
       \ref{lemma:NextFL}, we conclude $\av_2 = \av_2'$.
 \item Both $p$ and $q$ are respectively $\iif$ statements $\iif e_1$
       and $\iif e_2$.
       By the definition of $\udav$, we obviously have $\av_2 = \av_2'$.
       By Lemma \ref{lemma:NextC}, we have $p \not \in G^*(q)$ and $q
       \not \in G^*(p)$.
       By Lemma \ref{lemma:subgraphs_are_disjoint}, we have $G(p) \cap
       G(q) = \emptyset$.
       Therefore, we conclude $\ec_2=\ec_2'$.
\end{enumerate}
\end{proof}

By the above lemmata, we finally prove the following theorem, which
states that executions of a dPDG reach the same final sate regardless of
execution orders of statements.

\begin{theorem} \label{thm:dpdg}
 Let $s$ be a state of a run of a dPDG.
If there is a finite run $s \stackrel{n}{\to} s_1 \stackrel{n_1}{\to}
 \dots \stackrel{n_m}{\to} s_m$, all runs from state $s$ have the length
 $m+1$ and end with $s_m$.
\end{theorem}
\begin{proof}
 By induction on $m$.\\
\noindent \textbf{(Base)}
The claim is trivial for $m=0$.

\noindent \textbf{(Induction)}
Suppose that there is a run $s \stackrel{p}{\to} s_1 \stackrel{}{\to}
 \dots \stackrel{}{\to} s_m$ of the length $m+1$ starting with $s$.
Suppose $\Next(s) = \{p,p_0,\dots, p_j\}$.
By Lemma \ref{lemma:dPDG_non_interference}, we have $\{p_0,\dots,p_j\}
 \subseteq \Next(s_1)$.
Since we have a run $s_1 \stackrel{}{\to} \dots \stackrel{}{\to} s_m$
 starting from $s_1$ whose length is $m$, all runs from state $s_1$ have
 the length $m$ and end with $s_m$ by the induction hypothesis.
Therefore, for all $i \in [0,j]$, there is a finite run $s_1
 \stackrel{p_i}{\to} s_2^i \to \dots \to s_m$ of the length $m$.
Since we have $s_0 \stackrel{p}{\to} s_1 \stackrel{p_i}{\to} s_2^i$ by
 Lemma \ref{lemma:diamond_property}, there is $s_1^i$ such that $s_0
 \stackrel{p_i}{\to} s_1^i \stackrel{p}{\to} s_2^i$.
This implies there is a run starting with $s_1^i$, ending with $s_m$ and
 whose length is $m$.
By the induction hypothesis, all runs from state $s_1^i$ have the length
 $m$ and end with $s_m$.
Since $i$ is arbitrary, all runs from $s$ have the length $m+1$ and end
 with $s_m$.
\end{proof}

\subsection{Equivalence of the operational semantics of CFG and PDG}
Finally, we arrive at the equivalence of the operational semantics of
the CFG and the PDG.

\begin{theorem} \label{thm:cfgpdgsem}
 Let $P$ be a CFG and $G$ be the corresponding PDG.
Suppose $\sigma_0 \approx \av_0$ and $G$ is a dPDG.
If there is a finite run $\sigma_0 \stackrel{n_0}{\to} \sigma_1
 \stackrel{n_1}{\to} \dots \stackrel{n_m}{\to} \sigma_m$ of $P$ where
 $n_m = (\iret~x)$, there is a finite run $s_0 \stackrel{n_0'}{\to} s_1
 \stackrel{n_1'}{\to} \dots \stackrel{n_m'}{\to} s_m$ of $G$ and
 $\sigma_m(x) = \av_m(n_m, x)$.
\end{theorem}
\begin{proof}
 By Theorem \ref{thm:corresponding_run} and \ref{thm:dpdg}.
\end{proof}

\section{Conclusion} \label{sec:conclusion}
In this paper, we proposed an operational semantics of the PDG that
applies even to unstructured programs.
In our operational semantics, a PDG has the same execution sequence as
the corresponding CFG (Theorem \ref{thm:corresponding_run}).
We identified the class of deterministic PDGs (dPDGs) and proved that
every run of a dPDG starting with the same initial state end with the
same final state (Theorem \ref{thm:dpdg}).
These facts lead to the semantical equivalence between the CFG and the
PDG (Theorem \ref{thm:cfgpdgsem}).

We expect that our operational semantics will be used to formally
discuss correctness of program optimisation techniques based on PDG
transformations.
Since optimised PDGs should be translated into CFGs, there are some
algorithms that translate PDGs into CFGs
\cite{simons90foundation,steensgaard93sequentializing,zeng04generating}.
However, there are cases where optimised PDGs do not have corresponding
CFGs.
Thus, those algorithms duplicate some nodes or insert artificial
conditional statements.
Nevertheless, correctness of such transformations are not formally
proved.
One possible reason is that such PDGs are usually not well-structured
and the existing PDG semantics cannot be applied.
Since our semantics can be applied to such programs, it will be helpful
to prove the correctness of such algorithms.

Another interesting research direction is to prove semantical
equivalence between concurrent programs and PDGs.
This enables to formally discuss correctness of algorithms that
construct concurrent programs from PDGs.

\bibliographystyle{plain}
\bibliography{compiler}
\end{document}